%% file: root.tex
\documentclass[article]{IEEEtran}
\usepackage{cite}
\usepackage{color}
\usepackage[pdftex]{graphicx}
\usepackage{graphicx}
\graphicspath{{fig/}{jpeg/}}
\usepackage[cmex10]{amsmath}
\usepackage{amssymb}
\usepackage{algorithm}
\usepackage{algpseudocode}

\usepackage{amsmath,amssymb,lipsum}
\usepackage{hyperref}
\usepackage{comment}
\usepackage{graphicx}
\usepackage[font=small]{caption}
\usepackage{subcaption}
\usepackage{array}
\input{mysymbol.sty}
\input{my_sections.tex}

\usepackage{amsthm,bm}
\usepackage{tikz}
\usetikzlibrary{shapes,arrows}
\usepackage{pgfplots}
\usepgfplotslibrary{groupplots}
\usetikzlibrary{calc}
\usepgfplotslibrary{statistics}
\usepackage{pgfplotstable}
\pgfplotsset{compat=1.8}

\newtheorem{proposition}{Proposition}

\newtheorem{theorem}{Theorem}
\newtheorem{definition}{Definition}
\newtheorem{lemma}{Lemma}
\newtheorem{corollary}{Corollary}

\theoremstyle{definition}
\newtheorem{assumption}{Assumption}

\title{A Prediction-Correction Algorithm for Real-Time Model Predictive Control}
\author{Santiago Paternain, Manfred Morari and Alejandro Ribeiro
  \thanks{Work supported by ARL DCIST CRA W911NF-17-2-0181. The authors are with the Department of Electrical and Systems Engineering, University of Pennsylvania. Email: \{spater, morari, aribeiro\}@seas.upenn.edu.}
}

\begin{document}

\maketitle
\thispagestyle{empty}
\pagestyle{empty}

\input{introduction.tex}

\input{problem_formulation.tex}

\input{algorithm.tex}
\input{convergence.tex}
\input{stability.tex}
\input{numerical_examples.tex}
\input{conclusion.tex}

\input{appendix.tex}

\bibliographystyle{ieeetr}
\bibliography{bib}

\begin{IEEEbiography}{Santiago Paternain} 
received the B.Sc. degree in electrical engineering from Universidad de la Rep\'ublica Oriental del Uruguay, Montevideo, ´ Uruguay in 2012, the M.Sc. in Statistics from the Wharton School in 2018 and Ph.D. in Electrical and Systems Engineering from the Department of Electrical and Systems Engineering, the University of Pennsylvania in 2018. He was the recipient of the 2017 CDC Best Student Paper Award.  His research interests include optimization and control of dynamical systems.  
\end{IEEEbiography}

\begin{IEEEbiography}{Manfred Morari} (F’05) received the Diploma degree from ETH Zurich, Zurich, Switzerland, and the Doctoral degree from the University of Minnesota, Minneapolis, MN, USA, both in chemical engineering.

  He was the Head of the Department of Information Technology and Electrical Engineering at ETH Zurich from 2009 to 2012, and the Head of the Automatic Control Laboratory from 1994 to 2008. Prior to joining ETH Zurich, he was the McCollum-Corcoran Professor of chemical engineering and an Executive Officer for Control and Dynamical Systems at the California Institute of Technology. From 1977 to 1983, he was with the faculty of the University of Wisconsin. His interests include constrained and robust control, and his research is internationally recognized. The analysis techniques and software developed in his group are used in universities and industry throughout the world. He has received numerous awards, including the Eckman Award, Ragazzini Award, and Bellman Control Heritage Award from the American Automatic Control Council; the Colburn Award, Professional Progress Award, and CAST Division Award from the American Institute of Chemical Engineers; and the Control Systems Technical Field Award and the Bode Lecture Prize from IEEE. He is a Fellow of AIChE and IFAC. In 1993, he was elected to the U.S. National Academy of Engineering, and in 2015 to the UK Royal Academy of Engineering. He has also served on the technical advisory boards of several major corporations.
  
\end{IEEEbiography}

\begin{IEEEbiography}{Alejandro Ribeiro}  received the B.Sc. degree in electrical engineering from the Universidad de la Republica Oriental del Uruguay, Montevideo, in 1998 and the M.Sc. and Ph.D. degree in electrical engineering from the Department of Electrical and Computer Engineering, the University of Minnesota, Minneapolis in 2005 and 2007. From 1998 to 2003, he was a member of the technical staff at Bellsouth Montevideo. After his M.Sc. and Ph.D studies, in 2008 he joined the University of Pennsylvania (Penn), Philadelphia, where he is currently the Rosenbluth Associate Professor at the Department of Electrical and Systems Engineering. His research interests are in the applications of statistical signal processing to the study of networks and networked phenomena. His focus is on structured representations of networked data structures, graph signal processing, network optimization, robot teams, and networked control. Dr. Ribeiro received the 2014 O. Hugo Schuck best paper award, and paper awards at the 2016 SSP Workshop, 2016 SAM Workshop, 2015 Asilomar SSC Conference, ACC 2013, ICASSP 2006, and ICASSP 2005. His teaching has been recognized with the 2017 Lindback award for distinguished teaching and the 2012 S. Reid Warren, Jr. Award presented by Penn's undergraduate student body for outstanding teaching. Dr. Ribeiro is a Fulbright scholar class of 2003 and a Penn Fellow class of 2015.
\end{IEEEbiography}

\end{document}

%% file: my_sections.tex
\usepackage{needspace}



%% file: introduction.tex

%
\begin{abstract}
  In this work we adapt a prediction-correction algorithm for continuous time-varying convex optimization problems to solve dynamic programs arising from Model Predictive Control. In particular, the prediction step tracks the evolution of the optimal solution of the problem which depends on the current state of the system. The cost of said step is that of inverting one Hessian and it guarantees, under some conditions, that the iterate remains in the quadratic convergence region of the optimization problem at the next time step. These conditions imply (i) that the variation of the state in a control interval cannot be too large and that (ii) the solution computed in the previous time step needs to be sufficiently accurate. The latter can be guaranteed by running classic Newton iterations, which we term correction steps. Since this method exhibits quadratic convergence the number of iterations to achieve a desired accuracy $\eta$ is of order $\log_2\log_2 1/\eta$, where the cost of each iteration is approximately that of inverting a Hessian. This grants the prediction-correction control law low computational complexity. In addition, the solution achieved by the algorithm is such that the closed loop system remains stable, which allows extending the applicability of  Model Predictive Control to systems with faster dynamics and less computational power. Numerical examples where we consider nonlinear systems support the theoretical conclusions.  
\end{abstract}
%
%
\section{Introduction}
Model Predictive Control (MPC) is an optimal control technique that utilizes a process model to predict the future response of a plant and attempts to optimize its future behavior at each control interval. This technique has found acceptance in industrial applications due to its several advantages \cite{qin2003survey,garcia1989model}. Among these advantages, perhaps the most remarkable is the capability of imposing constraints to the state variables and the control inputs of the plant \cite{mayne2000constrained} by solving a dynamic program to decide the control action at each time step.
{Dynamic programs with infinite time horizon are intractable unless we are able to compute a closed form solution.} This is the case, for instance, when the objective function is quadratic in the state and the control inputs, the dynamics of the system are linear and there are no design constraints imposed on the system. These conditions however limit the range of applications in which these techniques can be used. To overcome this limitation, one can alternatively solve the receding horizon problem, where instead of minimizing an infinite sum, the focus is on the cumulative cost from the current time until a finite time horizon. This problem can be solved -- at least locally-- with classic optimization methods, e.g. Newton's method. The main drawback of this alternative is that at each control interval the controller needs to solve a different constrained optimization problem in order to select the optimal action. The latter prevents us from applying such techniques to systems with small time constants -- thus requiring small control intervals and short computation time -- or systems with small computation power. It is not surprising then, that MPC found traction first in the control of chemical processes where the previous conditions are typically met \cite{allgower1999nonlinear,biegler1991optimization,helbig1998model}. Improvements in computation algorithms for Model Predictive Control can allow expanding its range of applications to systems with faster dynamics and limited computing power such as small unmanned aerial vehicles \cite{alexis2011model}. 

Several efforts have been made to solve on-line and efficiently the MPC constrained optimization problem, see \cite{binder2001introduction} for a detailed comparison. These approaches are based on Newton-type methods that aim to solve exactly the receding horizon problem \cite{li1989multistep,de1995extension}. However, if the time for feedback is short, approximations are required. Some algorithms consider a one step horizon. Thus, reducing the number of optimization variables \cite{choi1993feedback,choi1999instantaneous}. Another common alternative is to linearize the system. This allows to solve the problem efficiently, because quadratic optimization problems with linear constraints can be solved with only one Newton step. Among these methods there are different approaches. Some linearize the system along a fixed optimal trajectory over the whole time horizon \cite{kramer1987numerical,kugelmann1990new}. Other approximation techniques perform successive linearizations along approximately optimal trajectories \cite{diehl2005real,diehl2005nominal,bock2000direct} and linearization of the dynamics of the system \cite{garcia1984quadratic,zheng1997computationally}. These solutions perform well as long as the system is not largely disturbed.

Instead of using approximations to reduce the complexity in the solution we propose to reduce the computation by tracking the solution to the receding horizon problem as it evolves with the state of the system. 
Specifically, we draw inspiration from a recent series of works,  where unconstrained \cite{simonetto2015prediction, simonetto2016class} and constrained \cite{fazlyab2016interior, fazlyab2017prediction} time-varying convex optimization problems have been considered. The general idea is to combine a prediction step that takes into account the temporal evolution of the optimization problem so that the predicted iterate does not drift away from the quadratic convergence neighborhood of the optimum. The latter makes possible to take advantage of the quadratic convergence of Newton's method -- or correction steps -- to obtain accurate solutions in a number of iterations that is of order $\log_2\log_2(1/\eta)$, where $\eta$ is the desired accuracy. 

The abovementioned prediction steps can be derived from the implicit function theorem, which has been used also in Model Predictive Control in \cite{zavala2009advanced}. In the latter the future state is predicted based on the model of the plant and a solution of the problem is computed. Once the state is observed the solution is corrected using a sensibility analysis based on the implicit function theorem. The main difficulty of this approach is that the computation of the predicted problem might be expensive since the seed of the optimization problem is not guaranteed to be in the quadratic convergence region. {In \cite{liao2018semismooth,liao2019semismooth} a semi-smooth predictor-corrector method is proposed to solve the Optimal control problem by tracking the roots of a parameterized non smooth root finding problem and sufficient conditions for asymptotic stability and constraint satisfaction are provided.}  

{
  Drawing inspiration from these algorithms, we consider the prediction-correction method proposed in \cite{paternain2019real} to efficiently solve on-line the time-varying problem that arises in MPC (Section \eqref{sec_algorithm}). The main differences with respect to the setting described in \cite{simonetto2015prediction, simonetto2016class,fazlyab2016interior, fazlyab2017prediction} are twofold. The first one is than in \cite{simonetto2015prediction, simonetto2016class,fazlyab2016interior, fazlyab2017prediction} the temporal evolution of the system\textemdash and thus of the optimal solution\textemdash is independent of the iterates that the prediction-correction method outputs. In the MPC setting, however, the solutions of the optimization problem are applied to the system as control inputs, and thus, future states -- hence the upcoming optimization problems -- depend on the output of the prediction-correction algorithm. The second difference is that the receding horizon optimization problem is not convex since the equality constraints describing the dynamics of the system that we are interested in controlling are generally non-linear. Despite these differences in the problem setting, the same theoretical guarantees in terms of the error of the solution can be established locally (Section \ref{sec_convergence}). That is,  (i) if the solution at a given control instant is sufficiently accurate then the error of the predicted iterate is at most of the order of the square of the sampling time (Propositions \ref{prop_optimum_prediction} and \ref{prop_error}) and therefore (ii) if this variation is not too large it is possible to ensure that the predicted iterate lies inside the quadratic convergence region (Proposition \ref{prop_qcr}). Building on these results we show that (iii) the prediction correction algorithm converges quadratically to a local solution of the receding horizon problem (Theorem \ref{theo_newton}). In addition, we establish a bound on the maximum state variation in a control interval that allows controller to use only two Hessian inversions (Corollary \ref{coro}).The conditions uner which (i) and (ii) hold are iimproved as compared to \cite{paternain2019real}. In addition, Section \ref{sec_stability} establishes that the error that results from solving the problem approximately is such that the system remains stable. Besides some conclusive remarks, the paper closes with numerical examples that support the theoretical results (Section \ref{sec_examples}).


%

%% file: problem_formulation.tex

\section{Model Predictive Control}\label{sec_problem_formulation}
In this work we are interested in reducing the computational effort of Model Predictive Control so we can apply it more easily in real time. Formally, let $\bbx \in \mathbb{R}^n$ denote the state of the system and $\bbu \in \mathbb{R}^p$ be its input. Then, the system of interest is described by a function $f:\mathbb{R}^n\times \mathbb{R}^p\to \mathbb{R}^n$ that relates the state and inputs at time $k$ to the state at time $k+1$
\begin{equation}\label{eqn_dynamics}
\bbx(k+1) = f(\bbx(k),\bbu(k)).
\end{equation}
Model predictive control selects the input of the system by solving an optimization problem that depends on the current state. The optimization problem predicts the model behavior of the system based on the dynamics \eqref{eqn_dynamics} over a given horizon $H$ and it attempts to minimize the cumulative cost 
\begin{equation}\label{eqn_finite_horizon}
J(k) = \sum_{l=k}^{k+H-1} \ell(\bbx(l),\bbu(l))+\ell_H(\bbx(k+H)), 
\end{equation}
where $\ell: \mathbb{R}^n\times \mathbb{R}^p \to \mathbb{R}$ and $\ell_H:\mathbb{R}^n\to\mathbb{R}$ are functions representing the performance metric of interest. The optimization variables representing the states and the input in future states are coupled through the dynamics of the system. Hence, in every control interval one is required to solve a constrained optimization problem. Formally, let us consider the optimization variables $\bar{\bbx}_{l}\in\mathbb{R}^n$ for $l=1,\ldots, H+1$ and $\bar{\bbu}_{l}\in\mathbb{R}^p$ for $l=1,\ldots, H$ and their concatenation $\bar{\bbx}\in \mathbb{R}^{n(H+1)}$ and $\bar{\bbu}\in \mathbb{R}^{pH}$ as  
\begin{equation}
    \bar{\bbx} = [\bar{\bbx}_{1}^\top, \bar{\bbx}_{2}^\top, \ldots, \bar{\bbx}_{H+1}^\top]^\top, \bar{\bbu} = [\bar{\bbu}_{1}^\top, \bar{\bbu}_{2}^\top, \ldots, \bar{\bbu}_{H}^\top]^\top.
\end{equation}
Then, the MPC control law selects as input of the system the vector $(\bar{\bbu}_1^\star)_k$ that arises from solving 
%
\begin{align}\label{eqn_optimization_problem}
({\bar{\bbx}}^\star,{\bar{\bbu}}^\star)_{(k)} : =&\argmin_{{\bar{\bbx}}\in\mathbb{R}^{n(H+1)},{\bar{\bbu}\in\mathbb{R}^{pH}}} J(\bbx(k),\bar{\bbx},\bar{\bbu}) \nonumber \\
  &\,\mbox{s.t.} \quad \, \bar{\bbx}_{i+1} -f(\bar{\bbx}_{i},\bar{\bbu}_i) = 0 \quad \forall i=1\ldots H\nonumber \\
  &\,\phantom{\mbox{s.t.}} \quad \, \bar{\bbx}_1 -\bbx(k) = 0.
\end{align}
Thus, the MPC law is $\kappa_{MPC}(\bbx(k)) : = \left(\bar{\bbu}_1^\star\right)_k$ and it yields the following closed loop dynamical system 
\begin{equation}\label{eqn_closed_loop_nominal}
\bbx(k+1) = f_{{MPC}}(\bbx(k)):= f(\bbx(k),\kappa_{MPC}(\bbx(k))).
\end{equation}
The origin of the previous system is guaranteed to be globally asymptotically stable under the typical assumptions.
\begin{assumption}[Continuity of system and cost]\label{assumption_continuity}
  The functions $f:\mathbb{R}^{n}\times \mathbb{R}^p \to \mathbb{R}^n$, $\ell:\mathbb{R}^n\times\mathbb{R}^p\to \mathbb{R}_+$ and $\ell_H:\mathbb{R}^n \to \mathbb{R}_{+} $ are continuous; $f(\bm{0},\bm{0})=\bm{0}$, $\ell(\bm{0},\bm{0})=0$ and $\ell_H(\bm{0})=0$.
\end{assumption}
\begin{assumption}\label{assumption_bounded_solution}
There exist compact sets $\ccalX\subset \mathbb{R}^{n(H+1)}$, $\ccalU \subset \mathbb{R}^{pH}$ such that for all $k\geq 0$ the solution $(\bar{\bbx}^\star,\bar{\bbu}^\star)_k$ to the problem \eqref{eqn_optimization_problem} satisfies $(\bar{\bbx}^\star,\bar{\bbu}^\star)_k\in\ccalX\times \ccalU$.
\end{assumption}
\begin{assumption}\label{assumption_nominal_stability}
    The stage $\ell(\cdot)$ and terminal cost $\ell_H(\cdot)$ satisfy
    \begin{equation}
    \ell(\bbx,\bbu)\geq \alpha_1(\left\|\bbx\right\|), \quad \ell_H(\bbx) \leq \alpha_2(\left\|\bbx\right\|),
  \end{equation}
in which $\alpha_1(\cdot)$ and $\alpha_2(\cdot)$ are $\ccalK_\infty$ functions and
\begin{equation}
\min_\bbu \ell_H(f(\bbx,\bbu))+\ell(\bbx,\bbu) \leq \ell_H(\bbx), \forall \bbx\in\mathbb{R}^n,
  \end{equation}
    \end{assumption}
%
\begin{theorem}\label{theo_mpc_stability}
Suppose that Assumptions \ref{assumption_continuity}, \ref{assumption_bounded_solution} and \ref{assumption_nominal_stability} are satisfied. Then the origin is globally asymptotically stable for the system \eqref{eqn_closed_loop_nominal}.
\end{theorem}
\begin{proof}
  See e.g. \cite[Theorem 2.24]{rawlings2009model}.
  \end{proof}
%
The main difficulty when applying Model Predictive Control is the computation of an accurate solution in the control interval. Note that at each time-step $k$, one needs to solve a new optimization problem of the form \eqref{eqn_optimization_problem} that depends on the current state of the system. In general, the algorithms used to do so are based on Newton's method whose cost per iteration is approximately that of inverting a Hessian. Unless the dynamics of the system are linear the computation of each control action requires several steps of the algorithm and thus, it might be the case that at the time of actuation a sufficiently accurate solution to \eqref{eqn_optimization_problem} has not been computed yet. The latter is the case especially when the systems have low computational power or small time constants.  

To reduce the computational cost of solving the MPC problem \eqref{eqn_optimization_problem} we propose to exploit the quadratic convergence rate that Newton's method exhibit in the neighborhood of the solutions (see e.g. \cite[Section 9.5.3]{boyd2004convex}). In particular we propose to use a prediction-step, whose computational cost is also that of inverting a Hessian, so to guarantee that the seed used to solve the problem at every time step lies in the quadratic convergence region. Hence reducing the computational cost of the overall algorithm as compared to classic Newton's method which might operate in the damped phase before reaching the quadratic convergence region. Despite the low complexity of the prediction-correction algorithm proposed, the solutions are such that the closed loop system remains stable.


%% file: algorithm.tex

\section{Prediction-Correction PC-MPC }\label{sec_algorithm}
Optimization algorithms based on Newton's method are guaranteed to find local solutions to receding  horizon problem \eqref{eqn_optimization_problem}. To characterize these solutions, we start by writing the Lagrangian associated with it
\begin{equation}
  \begin{split}
    \ccalL(\bbx(k),\bar{\bbx},\bar{\bbu},\bm{\lambda}) = J(\bbx(k),\bar{\bbx},\bar{\bbu})+\bm{\lambda}_1^\top\left(\bar{\bbx}_1- \bbx(k)\right) \\+ \sum_{i=2}^{H+1} \bblambda_i^\top\left(\bar{\bbx}_i-f(\bar{\bbx}_{i-1},\bar{\bbu}_{i-1})\right).
    \end{split}
  \end{equation}
where $\bm{\lambda}\in\mathbb{R}^{n(H+1)}$ is a vector containing the multipliers corresponding to the $H+1$ constraints. To simplify the notation, we  define the following vector containing all the optimization variables $\bbz=\left[\bar{\bbx}^\top, \bar{\bbu}^\top, \bm{\lambda}^\top\right]^\top \in \mathbb{R}^{H(2n+p)+2n}$ and we write compactly $\ccalL(\bbx(k),\bbz)$. The first order condition for optimality of problem \eqref{eqn_optimization_problem} is given by the gradient of the Lagrangian with respect to $\bbz$ being equal to zero and these are the solutions that Newton based methods are guaranteed to find. Hence, we also settle for local solutions of the abovementioned form, i.e., 
\begin{equation}\label{eqn_first_order_opt}
\nabla_{\bbz} \ccalL(\bbx(k),\bbz^\star_k) = \mathbf{0}.
\end{equation}
Let us denote by $\bbz_{k+1}^0$ the seed used to solve problem \eqref{eqn_optimization_problem} at time $k+1$, and let us write a first order Taylor expansion of $\nabla_{\bbz}\ccalL(\bbx(k+1),\bbz_{k+1}^0)$ around the pair $(\bbx(k),\bbz_k^\star)$.  Since this pair satisfies \eqref{eqn_first_order_opt}, it follows that  
\begin{align}\label{eqn_taylor}
    \nabla_{\bbz} \ccalL(\bbx(k+1),&\bbz^0_{k+1}) \approx \nabla_{\bbz \bbx} \ccalL(\bbx(k),\bbz^\star_k)\left(\bbx(k+1)-\bbx(k)\right) \nonumber \\
    &+ \nabla_{\bbz \bbz} \ccalL(\bbx(k),\bbz^\star_k)\left(\bbz_{k+1}^0-\bbz_k^\star\right).
  \end{align}
Recall that Newton's method exhibits quadratic convergence in the neighborhood of the critical points (see e.g. \cite[Section 9.5.3]{boyd2004convex}) defined as
\begin{equation}\label{eqn_qcr}
  QCR(\bbz^\star_k) = \left\{\bbz\in\mathbb{R}^{H(2n+p)+2n}\mid \left\|\bbz-\bbz_k^\star\right\|\leq \frac{m}{L}\right\}.
\end{equation}
The importance of this fact, is that to achieve a desired accuracy $\varepsilon$ on the solution of \eqref{eqn_first_order_opt}, we require only $N = O(\log \log(1/\varepsilon))$ Newton iterations. Thus, to achieve good accuracy of the solution in few iterations it is of interest that the selected seed lies in said region. This means that the seed needs to be such that the norm of the gradient in \eqref{eqn_taylor} is close to zero.  The latter is achieved \textemdash to a first order approximation \textemdash using a predicted seed of the form
\begin{equation}\label{eqn_implicit_function_theorem}
\bbz_{k+1|k}^\star = \bbz_k^\star-\nabla_{\bbz \bbz} \ccalL_{\star}(k)^{-1} \nabla_{\bbz \bbx} \ccalL_{\star}(k)\left(\bbx(k+1)-\bbx(k)\right),
\end{equation}
where the notation $\ccalL_{\star}(k)$, introduced for simplicity, denotes the evaluation of $\ccalL(\cdot,\cdot)$ at $(\bbx(k),\bbz^\star_k)$. The latter equation suggests that a similar update could be applied to $\bbz_k$ \textemdash the approximation of the solution to \eqref{eqn_optimization_problem} available at time $k$ \textemdash to track the local solution.  We term this update the prediction step and it selects the seed $\bbz_{k+1}^0$ as 
\begin{equation}\label{eqn_prediction}
\bbz_{k+1}^0 = \bbz_k-\nabla_{\bbz \bbz} \ccalL({k})^{-1} \nabla_{\bbz \bbx} \ccalL({k})\left({\bbx}(k+1)-\bbx(k)\right),
\end{equation}
where we have defined $\nabla_{\bbz \bbz} \ccalL({k}):= \nabla_{\bbz \bbz} \ccalL(\bbx(k),\bbz_k)$ and $\nabla_{\bbz \bbx} \ccalL({k}):=\nabla_{\bbz \bbx} \ccalL(\bbx(k),\bbz_k) $. Since the predicted iterated is based on the first order Taylor expansion (cf., \eqref{eqn_implicit_function_theorem}) it is expected that the smaller the variation of the system's state and the closer $\bbz_k$ is to $\bbz_k^\star$ \textemdash the solution to problem \eqref{eqn_optimization_problem}\textemdash the closer the prediction $\bbz_{k+1}^0$ is to  $\bbz_{k+1}^\star$. We formalize this intuition in Proposition \ref{prop_error}. Under some assumptions on the variation of the state in a control interval we can use said result to guarantee that the predicted iterate $\bbz_{k+1}^0$ lies in the quadratic convergence region of the problem \eqref{eqn_optimization_problem} at time $k+1$ (Proposition \ref{prop_qcr}). That being the case we can run Newton's method with step size one. We term this update, the correction step 
\begin{equation}\label{eqn_correction}
  \begin{split}
    &\bbz_{k+1}^{j+1}= \bbz_{k+1}^j-\nabla^2_{\bbz\bbz} \ccalL_{j}(k+1)^{-1} \nabla_\bbz \ccalL_j(k+1),
    \end{split}
  \end{equation}
where we have defined $\nabla_{\bbz \bbz} \ccalL_j({k}):= \nabla_{\bbz \bbz} \ccalL(\bbx(k),\bbz_k^j)$ and $\nabla_{\bbz \bbx} \ccalL_j({k}):=\nabla_{\bbz \bbx} \ccalL(\bbx(k),\bbz_k^j) $. As previously mentioned the fact that the seed lies in the quadratic convergence region depends on a bound on the state variation between consecutive time steps and the accuracy of the solution computed in the previous time step. Hence, by allowing multiple correction steps \textemdash which allows for a more accurate solution \textemdash  we are trading-off computational cost for the ability of controlling systems that vary more. We formalize this trade-off in Theorem \ref{theo_newton} where we bound the number of correction steps required to ensure that the seed lies in the quadratic convergence region for all $k$. The prediction-correction scheme \eqref{eqn_prediction}--\eqref{eqn_correction} defines the Prediction-Correction MPC (PC-MPC) law
\begin{equation}\label{eqn_control_law}
  \bbu(k) = \kappa(\bbx(k),\bbx(k-1),\bbz_{k-1}):= \left(\bar{\bbu}_{1}^N\right)_k, \quad \forall k>0,
\end{equation}
where $\left(\bar{\bbu}_{1}^N\right)$ represents the iterate after $N$ corrections of the form \eqref{eqn_correction} to solve the problem \eqref{eqn_optimization_problem} at time $k$, when the seed has been selecting using the prediction step \eqref{eqn_prediction}. This control law defines the following closed loop dynamical system when applied as the input of system \eqref{eqn_dynamics}
\begin{equation}\label{eqn_closed_loop}
  \begin{split}
    \bbx(k+1) = f_{\kappa}(\bbx(k),\bbx(k-1),\bbz_{k-1})\\
    := f(\bbx(k),\kappa(\bbx(k),\bbx(k-1),\bbz_{k-1})). 
    \end{split}
  \end{equation}
We summarize the feedback loop with the prediction correction algorithm under Algorithm \ref{alg_pcmpc}. Note that the initial control input $\bbu(0)$ can be computed by solving \eqref{eqn_optimization_problem} with classic Newton steps. Since this computation can be done off line before the system starts evolving, the complexity required to ensure an accurate solution of \eqref{eqn_optimization_problem} is affordable.
%
\begin{algorithm}
  \caption{predictionCorrectionMPC}
  \label{alg_pcmpc} 
\begin{algorithmic}[1]
 \renewcommand{\algorithmicrequire}{\textbf{Input:}}
 \renewcommand{\algorithmicensure}{\textbf{Output:}}
 \Require $\bbx(0),N,\varepsilon$
 \State Compute $\bbz_0=\bbz^\star(\bbx(0))$ 
 \For{$k=0,1,\ldots$}
 \State Apply input $\bbu(k) = \left(\bbu_1\right)_k$ to the system and observe
 $$\bbx(k+1)=f(\bbx(k),\bbu(k))$$
 \State Compute prediction step $\bbz_{k+1}^0$ according to \eqref{eqn_prediction}
 $$
\bbz_{k+1}^0 = \bbz_k-\nabla_{\bbz \bbz} \ccalL(k)^{-1} \nabla_{\bbz \bbx} \ccalL(k)\left(\bbx(k+1)-\bbx(k)\right)$$
 \State Set $j=0$
 \While {$j<N$ or $\left\|\nabla_\bbz\ccalL(\bbx(k+1),\bbz_{k+1}^j)\right\|>\varepsilon$}
 \State $j=j+1$
 \State Compute Correction (or Newton) step as in \eqref{eqn_correction}
 $$
 \bbz_{k+1}^{j}= \bbz_{k+1}^{j-1}-\nabla^2_{\bbz\bbz} \ccalL_{j-1}(k+1)^{-1} \nabla_\bbz \ccalL_{j-1}(k+1),
 $$
 \EndWhile
 \State Update variable $ \bbz_{k+1} = \bbz_{k+1}^j$
 \EndFor
%
 \end{algorithmic}
 \end{algorithm}
%

Recall that the prediction step relies on a first order approximation and thus for it to efficiently track the optimal solution we require that the system does not evolve arbitrarily fast. In the next assumption we impose a bound on the norm of the difference between states on two consecutive times.
%
\begin{assumption}\label{assumption_bound_diff_states}
Let $\bbz_0$ be an approximate solution of \eqref{eqn_optimization_problem} at time $k=0$. 
Denote by $k$ the time index, $\bbx(0)$ the initial condition of the system and $\bm{\phi}_{\kappa}(k,\bbx(0),\bbz_0)$ the solution of the dynamical system \eqref{eqn_closed_loop}.  Then, for all $k\in\mathbb{N}$ and for all $\bbx(0)\in \mathbb{R}^n$ there exists a constant $B>0$ such that 
  \begin{equation}\label{eqn_condition_sampling_time}
\begin{split}
  \left\|\bm{\phi}_{\kappa}(k+1,\bbx(0),\bbz_0) - \bm{\phi}_{\kappa}(k,\bbx(0),\bbz_0)\right\| 
\leq {B}T_s,
\end{split}
  \end{equation}
  where $T_s$ is the sampling time of the system.  
  \end{assumption}
The previous assumption imposes a bound on the variation of the state in consecutive times. This can be interpreted as if the continuous time counterpart of the dynamical system had derivatives bounded by ${B}$, then the variation of the solution in a control interval would have to be smaller than ${B}T_s$.

In addition to the bound in the variation of the states we are required to guarantee that $\nabla^2_{\bbz \bbz} \ccalL(\bbx(k), \cdot)$ is invertible for all times $k\geq 0$ since both the prediction \eqref{eqn_prediction} and the correction step \eqref{eqn_correction} rely on its inverse. The later can be ensured by the Sufficient Second Order Conditions and the Linear Independence Constraint Qualifications, i.e., the Hessian of the Lagrangian evaluated at $(\bbx,\bbz^\star(\bbx))$ is positive definite in any feasible direction and that the gradients of the constraints evaluated at any feasible point are linearly independent \cite{nocedal2006numerical}. In this work, however, we make a stronger requirement, a uniform bound on the absolute value of the eigenvalues. These requirements are standard in Newton type analysis \cite[Section 9.5.3]{boyd2004convex} and they allow to establish convergence rates. In convex optimization, this bound is enforced by the strong convexity assumption necessary to have quadratic convergence of Newton's method. We require as well other smoothness conditions on the functions standard in the prediction-correction literature \cite{simonetto2016class,fazlyab2017prediction}. We formally state these assumptions next.
%
\begin{assumption}\label{assumption_non_zero_eigenvalues}
  Let $\bbz^\star:\mathbb{R}^n \to \mathbb{R}^{(n+p+2)H}$ be such that
  \begin{equation}
    \nabla_\bbz \ccalL(\bbx,\bbz^\star(\bbx)) =0.
  \end{equation}
  For all $\bbx \in \mathbb{R}^n$, denote by $\lambda_{i}\left(\nabla^2_{\bbz\bbz}\ccalL(\bbx,\bbz^\star(\bbx))\right)$, with $i=1,\ldots, (n+p+2)H$, the eigenvalues of the second derivative of the Lagrangian with respect to $\bbz$ at the point $(\bbx,\bbz^\star(\bbx))$. We assume that there exists a uniform bound $m>0$ such that for all $\bbx$ we have
  \begin{equation}
    \min_{i=1\ldots (n+p+2)H} \left|\lambda_{i}\left(\nabla^2_{\bbz\bbz}\ccalL(\bbx,\bbz^\star(\bbx))\right) \right|> 2m.
   \end{equation}
\end{assumption}
\begin{assumption}\label{assumption_bounded_derivatives}
  The function $\ccalL(\bbx,\bbz)$ is sufficiently smooth both in $\bbx \in \mathbb{R}^n$ and $\bbz \in \mathbb{R}^{H(n+p+2)}$. In particular the second partial derivative with respect to $\bbz$ and $\bbx$ is bounded
  \begin{equation}
  \left\| \nabla^2_{\bbz\bbx}\ccalL(\bbx,\bbz)\right\| \leq C.
  \end{equation}
  In addition we assume that the derivative of $\nabla_\bbz\ccalL(\bbx,\bbz)$ is L-Lipschitz. Let $D(\bbx,\bbz) = \left[\nabla^2_{\bbz\bbz}\ccalL(\bbx,\bbz), \nabla^2_{\bbz\bbx}\ccalL(\bbx,\bbz) \right] $, then, there exists a positive constant $L$ such that for all $\bby_1 = [\bbx_1^\top,\bbz_1^\top]^\top$ and $\bby_2 =[\bbx_2^\top,\bbz_2^\top]^\top$ it holds that 
  \begin{equation}
    \left\| D(\bbx_1,\bbz_1)-D(\bbx_2,\bbz_2)\right\| \leq L\left\|\bby_1-\bby_2 \right\|.
    \end{equation}
  In particular, the Lipschitz assumption implies also that 
  \begin{equation}
    \left\| \nabla^2_{\bbz\bbz}\ccalL(\bbx,\bbz_1)-\nabla^2_{\bbz\bbz}\ccalL(\bbx,\bbz_2)\right\| \leq L\left\|\bbz_1-\bbz_2 \right\|.
    \end{equation}
  %
  %
\end{assumption}

%
Note that, since we assume that the Lagrangian is differentiable, the dynamics of the system and the objective functions need to be so as well. Hence, Assumption \ref{assumption_bounded_derivatives} implies the continuity in Assumption \ref{assumption_continuity}. In the next section we establish the accuracy with which the optimization problem needs to be solved at time $k$ for the predicted iterate $\bbz_{k+1}^0$ to be in the quadratic convergence region of problem \eqref{eqn_optimization_problem} at time $k+1$ (Proposition \ref{prop_qcr}). In addition, we establish a bound on the maximum number of correction steps that are required to achieve the aforementioned accuracy (Theorem \ref{theo_newton}).

%% file: convergence.tex

\section{Computational Complexity of PC-MPC}\label{sec_convergence}
We start the analysis of the computational complexity of the proposed algorithm by formalizing the bound on the norm of the difference of the predicted seed $\bbz_{k+1}^0$ and $\bbz_{k+1}^\star$, the solution of the optimization problem \eqref{eqn_optimization_problem} at time $k+1$. To do so, we require an intermediate result that bounds the error of the predicted seed starting at the solution to the optimization problem at time $k$. To be precise, let $\left\{\bbz_k^\star, k \geq 0\right\}$ be the sequence of solutions to problem \eqref{eqn_optimization_problem} when the control law applied to the system \eqref{eqn_dynamics} is \eqref{eqn_control_law}. Based on the expression \eqref{eqn_implicit_function_theorem} define the predicted optimal sequence $\left\{\bbz^\star_{k+1|k}, k\geq 0 \right\}$
\begin{equation}\label{eqn_optimum_prediction}
  \begin{split}
  \bbz^\star_{k+1|k} = \bbz^\star_{k}-\nabla^2_{\bbz\bbz}\ccalL_{\star}(k)^{-1}\nabla^2_{\bbx\bbz}\ccalL_{\star}(k)\Delta \bm{\phi}_{\kappa}(k).
\end{split}
\end{equation}
To simplify the notation in the previous expression, we have defined $\nabla^2_{\bbz\bbz}\ccalL_\star(k):=\nabla^2_{\bbz\bbz}\ccalL(\bm{\phi}_k(k,\bbx(0),\bbz_0),\bbz_{k}^\star)$, $\nabla^2_{\bbx\bbz}\ccalL_\star(k):=\nabla^2_{\bbx\bbz}\ccalL(\bm{\phi}_k(k,\bbx(0),\bbz_0),\bbz_k^\star)$ and $ \Delta \bm{\phi}_{\kappa}(k) := \bm{\phi}_k(k+1,\bbx(0),\bbz_0)-\bm{\phi}_k(k,\bbx(0),\bbz_0)$. Because $\bbz_{k+1|k}^\star$ is a first order approximation of the dependence of the optimal solution with the state, it is expected that the difference between $\bbz^\star_{k+1|k}$ and the optimal solution $\bbz_{k+1}^\star$ of problem \eqref{eqn_optimization_problem} at time $k+1$ is bounded by $O(B^2T_s^2)$ as we formally state in the next proposition. 
\begin{proposition}\label{prop_optimum_prediction}
Let $\bm{\phi}_k(k,\bbx(0),\bbz_0)$ be the solution of the dynamical system \eqref{eqn_closed_loop} with initial conditions $\bbx(0)$, $\bbz_0$, let $\{\bbz_k^\star, k\geq 0\}$ be the sequence of solutions to the problem \eqref{eqn_optimization_problem} for the trajectory $\bm{\phi}_k(k,\bbx(0),\bbz_0)$, and let $\left\{\bbz^\star_{k+1|k}, k\geq 0\right\}$ be the sequence defined in \eqref{eqn_optimum_prediction}. Under Assumptions \ref{assumption_bound_diff_states}--\ref{assumption_bounded_derivatives}, for all $k \geq 0$ we have that 
  \begin{equation}
    \left\|  \bbz_{k+1}^\star -\bbz_{k+1|k}^\star\right\| \leq \frac{L}{m}\delta_1B^2T_s^2,
  \end{equation}
  where $B$ is the constant defined in Assumption \ref{assumption_bound_diff_states}, $T_s$ is the sampling time of the system and
    \begin{equation}\label{eqn_delta1}
    \delta_1:=\frac{1}{4}\left(\frac{C}{2m}+1\right)\left(\frac{C^2}{4m^2}+1\right)^{1/2},
  \end{equation}  
    with $m, C$ and $L$ being the constants defined in assumptions \ref{assumption_non_zero_eigenvalues} and \ref{assumption_bounded_derivatives}.
      
\end{proposition}
\begin{proof}
See Appendix \ref{app_optimum_prediction}.
\end{proof}}
The previous result establishes a bound on the error that the prediction step introduces. Since the prediction step is a first order approximation of the function $\bbz^\star(\bbx)$ is not surprising that the error is of the order of the square of the variation of the state in a control interval. The previous result also establishes a bound on the maximum variation allowable for the predicted iterate to lie in quadratic convergence region. Since the quadratic convergence region is given by the points satisfying $\left\|\bbz-\bbz_k^\star\right\| \leq m/L$ (cf., \eqref{eqn_qcr}) we require that $B^2 T_s^2 \leq m^2\left(L^2\delta_1\right)^{-1}$ to ensure that $\bbz_{k+1|k}^\star$ is in the region. The next result uses the tracking error established to bound the error of the prediction step when $\bbz_k\in QCR(\bbz_k^\star)$ instead of being the exact solution.
%
\begin{proposition}\label{prop_error}
  Consider $\bm{\phi}_k(k,\bbx(0),\bbz_0)$, the solution of the dynamical system \eqref{eqn_closed_loop} with initial conditions $\bbx(0)$ and $\bbz_0$. Let Assumptions \ref{assumption_bound_diff_states}--\ref{assumption_bounded_derivatives} hold and let $\{\bbz_k^\star, k\geq 0\}$ be the sequence of solutions to problem \eqref{eqn_optimization_problem} for the trajectory $\bm{\phi}_k(k,\bbx(0),\bbz_0)$. Let $m,C$ and $L$ be the constants defined in Assumption \ref{assumption_non_zero_eigenvalues} and \ref{assumption_bounded_derivatives}, $\delta_1$ be the constant in \eqref{eqn_delta1} and define the following constant 
  \begin{equation}\label{eqn_delta2}
    \delta_2 :=\frac{1}{2}\left(\frac{C}{m}+1\right).
    \end{equation}
If  
$\left\| \bbz_k-\bbz_k^\star\right\|<m/L$, for some $k\geq 0$, then we have that  
  \begin{equation}\label{eqn_bound_prop_error}
    \begin{split}
      \left\|\bbz_{k+1}^0-\bbz^\star_{k+1}\right\| \leq \frac{L}{m}\delta_1B^2T_s^2
      +\left(1+\frac{L}{m}\delta_2{B}T_s\right)\left\|\bbz_{k}-\bbz^\star_{k}\right\|. 
    \end{split}
    \end{equation}
\end{proposition}
\begin{proof}
See Appendix \ref{app_prop_error}.
  \end{proof}
%
%
The bound of the error of the predicted step is a function that depends on ${B}T_s$ which is measure of how much the system is allowed to change and $\left\|\bbz_k-\bbz_k^\star\right\|$, which is how good is our solution for the problem at time $k$. The slower the system is, and the more accurate the solution at time $k$ is, the smaller the tracking error. And thus, the easier it is to guarantee that the predicted step $\bbz_{k+1}^0$ is in the quadratic convergence region of the solution of \eqref{eqn_optimization_problem} at time $k$. These two quantities define a trade-off that relates the speed of the system of interest and the computational effort that is required to control it. In particular, the next proposition establishes a bound on the accuracy required in the solution of \eqref{eqn_optimization_problem} in order to guarantee that $\bbz_{k+1}^0$ is indeed in the quadratic convergence region depending on the state's variation in a control interval. 
%
%
\begin{proposition}\label{prop_qcr}
  Let Assumptions \ref{assumption_bound_diff_states}--\ref{assumption_bounded_derivatives} hold. Further let $\delta_1, \delta_2$  be the constants defined in \eqref{eqn_delta1}, \eqref{eqn_delta2} and let $L$ and $m$ be the constants in assumptions \ref{assumption_non_zero_eigenvalues} and \ref{assumption_bounded_derivatives}. If the variation between consecutive states satisfy $B^2T_s^2<m^2(L^2\delta_1)^{-1}$, then, $\bbz_{k+1}^0$, computed as in \eqref{eqn_prediction}, is in the quadratic convergence region of problem \eqref{eqn_optimization_problem} at time $k+1$, if the problem at time $k$ has been solved with accuracy, at least $\eta m/L$, where $\eta$ satisfies 
  \begin{equation}\label{eqn_eta_inequality}
    \eta\leq \frac{1-\delta_1\frac{L^2}{m^2} B^2T_s^2}{1+\delta_2\frac{L}{m}{B}T_s}.
  \end{equation}
\end{proposition}
\begin{proof}
  Using the result of Proposition \ref{prop_error} and the assumption that  $\left\|\bbz_k-\bbz_k^\star\right\|\leq \eta m/L$, it follows that 
  \begin{equation}
      \left\|\bbz_{k+1}^0-\bbz^\star_{k+1}\right\| \leq \frac{L}{m}\delta_1B^2T_s^2
      +\left(1+\frac{L}{m}\delta_2{B}T_s\right)\frac{\eta m}{L}.
  \end{equation}
  Recall that the quadratic convergence region is given by all the points that are at a distance smaller than $m/L$ (cf., \eqref{eqn_qcr}), this is
  \begin{equation}
    \left\|\bbz-\bbz_{k+1}^\star\right\| \leq \frac{m}{L}.
  \end{equation}
 Which implies that the seed $\bbz_{k+1}^0$ is in the quadratic convergence region if
  \begin{equation}
   \delta_1 \frac{L^2}{m^2}  B^2T_s^2+\delta_2\frac{L}{m}{B}T_s\eta -(1-\eta) \leq 0.
  \end{equation}
  Solving for $\eta$ yields \eqref{eqn_eta_inequality}. This completes the proof.
\end{proof}
The previous proposition establishes a trade-off between how fast the system is allowed to vary and how accurate the solution of the problem at time $k$ needs to be for the predicted iterate $\bbz_{k+1}^0$ to be in the quadratic convergence region. If one were to solve the problem exactly, i.e. with $\eta =0$, the previous bound reduces to ${B}T_s\leq m/(L\sqrt{\delta_1})$ as discussed after Proposition \ref{prop_optimum_prediction}. On the other hand if one were to solve the problem with the minimum allowed accuracy, i.e., $\eta=1$ we would require the system to satisfy ${B}T_s = 0$, i.e., not change over time. In summary, slower variations of the systems allow for less accuracy in the solution. 

Using the fact that the predicted iterate can be placed inside the quadratic convergence region by finding a solution at time $k$ that is sufficiently accurate \textemdash under the hypothesis of Proposition \ref{prop_qcr} \textemdash we set focus in deriving a bound on the number of iterations required to do so. This bound is a function of order $\log (\log(\eta))$, where $\eta$ satisfies \eqref{eqn_eta_inequality}. Thus, suggesting the low computational cost of the overall scheme. This result follows from the analysis of the Newton step in the quadratic convergence region as we formalize next. 
%
\begin{theorem}\label{theo_newton}
  Let $\bm{\phi}_k(k,\bbx(0),\bbz_0)$ be the solution of the dynamical system \eqref{eqn_closed_loop}, with initial conditions $\bbx(0)$ and $\bbz_0$. Let Assumptions \ref{assumption_bound_diff_states}--\ref{assumption_bounded_derivatives} hold and consider the sequence $\{\bbz_k^\star, k\geq 0\}$ of solutions to problem \eqref{eqn_optimization_problem} for the trajectory $\bm{\phi}_k(k,\bbx(0),\bbz_0)$. Let $\delta_1$ and $\delta_2$ be the constants defined in \eqref{eqn_delta1} and \eqref{eqn_delta2}  and assume that the difference between states of the system in two consecutive time steps is such that $B^2T_s^2<m^2(L^2\delta_1)^{-1}$. Then, if $\left\|\bbz_0-\bbz_0^\star\right\| \leq \eta m /L$, the prediction step \eqref{eqn_prediction} along with $N$ correction steps \eqref{eqn_correction}, with $N$ satisfying 
%
  %
  \begin{equation}\label{eqn_number_corrections}
    N \leq \log_2\left(1+\log_2\frac{1+\delta_2{B}T_s\frac{L}{m}}{1-\frac{L^2}{m^2}\delta_1 B^2T_s^2}\right),
  \end{equation}
  ensures that $\left\|\bbz_k-\bbz_k^\star\right\|\leq \eta m/L$ for all $k\geq 0$.
  %
%
  %
\end{theorem}
\begin{proof}
  The proof is inspired by that in \cite{simonetto2016class}. Let us show the result by induction. The statement holds for $k=0$ by hypothesis. Hence, we are left to show that if $\left\|\bbz_k-\bbz_k^\star\right\| \leq \eta m/L$, the same holds for $k+1$. In Proposition \eqref{prop_qcr} we have established that the predicted seed is such that $\bbz_{k+1}^0 \in QCR(\bbz_{k+1}^\star)$ under the assumption that $\left\|\bbz_k-\bbz_k^\star\right\| \leq \eta m/L$. Thus, we just need to show that $N$ correction steps with $N$ bounded as in \eqref{eqn_number_corrections} yield a solution with accuracy $\eta m/L$. Hence, we proceed to analyze the correction steps \eqref{eqn_correction}. Let us write the difference $\bbz_{k+1}^{j+1}-\bbz_{k+1}^\star$ as 
\begin{equation}\label{eqn_newton_1}
  \begin{split}
  \bbz_{k+1}^{j+1}-\bbz_{k+1}^\star = \bbz_{k+1}^j-\bbz_{k+1}^\star - \nabla^2_{\bbz\bbz}\ccalL(\bbz_{k+1}^j)^{-1}\nabla_\bbz\ccalL(\bbz_{k+1}^j),
 \end{split}
  \end{equation}
where we have dropped the state of the system $\bm{\phi}_{\kappa}(k+1,\bbx(0))$ in order to simplify the notation.  For the same reasons, define $\Delta \bbz = \bbz_{k+1}^\star-\bbz_{k+1}^{j}$. Then, using the Fundamental Theorem of Calculus along with the fact that $\nabla_{\bbz}\ccalL(\bbz_{k+1}^\star) = \bm{0}$ it follows that
\begin{equation}
\nabla_{\bbz}\ccalL(\bbz_{k+1}^{j}) = -\int_{0}^1\nabla_{\bbz\bbz}^2\ccalL(\bbz_{k+1}^{j}+\theta\Delta\bbz)\Delta\bbz d\theta.
\end{equation}
Then, we can write \eqref{eqn_newton_1} as 
\begin{equation}\label{eqn_newton_2}
  \begin{split}
  \bbz_{k+1}^{j+1}-\bbz_{k+1}^\star = \nabla^2_{\bbz\bbz}\ccalL(\bbz_{k+1}^j)^{-1}\\\int_{0}^1\left(\nabla_{\bbz\bbz}^2\ccalL(\bbz_{k+1}^j+\theta\Delta\bbz)-\nabla_{\bbz\bbz}^2\ccalL(\bbz_{k+1}^\star\right)\Delta\bbz d\theta,
 \end{split}
  \end{equation}
Using the Lipschitz continuity of $\nabla_{\bbz\bbz}\ccalL(\bbx,\bbz)$ (cf., Assumption \ref{assumption_bounded_derivatives}) and the fact that its minimum eigenvalue is larger than $m$ for any $\bbz \in QCR(\bbz_{k+1}^\star)$ (cf., Lemma \ref{lemma_bound_gradient}), we can upper bound the norm $\left\|\bbz_{k+1}^{j+1}-\bbz_{k+1}^\star\right\|$ as 
\begin{equation}
  \begin{split}
    \left\|\bbz_{k+1}^{j+1}-\bbz_{k+1}^\star\right\| \leq \frac{1}{m}\int_0^1 L\theta\left\|\Delta\bbz\right\|^2 d\theta.
    \end{split}
\end{equation}
Substituting $\Delta\bbz$ by its definition and integrating $\theta$ between $0$ and $1$ yields the following bound establishing quadratic convergence of the Newton step
 \begin{equation}
\left\|\bbz_{k+1}^{j+1}-\bbz_{k+1}^\star\right\| \leq \frac{L}{2m}\left\|\bbz_{k+1}^j-\bbz_{k+1}^\star\right\|^2.
  \end{equation}
  Applying the recursion it is possible to bound the norm of the $N-th$ correction step at the iterate $k+1$ by 
  \begin{equation}
\frac{L}{m}\left\|\bbz^{N}_{k+1}-\bbz_{k+1}^\star\right\| \leq 2^{-(2^N-1)}\left(\frac{L}{m}\left\|\bbz_{k+1}^0-\bbz_{k+1}^\star\right\|\right)^{(2N)}.
  \end{equation}
  Since we have that the prediction step is in the quadratic convergence region it follows that ${L}/{m}\left\|\bbz_{k+1}^0-\bbz_{k+1}^\star\right\|\leq 1$. Thus, if one requires to solve the problem with accuracy $\eta$, the number of iterations required is
  \begin{equation}
    N \leq  \log_2\left(1+\log_2(1/\eta))\right),
    \end{equation}
  To complete the proof of the Theorem, use the result in Proposition \ref{prop_qcr} that establishes the accuracy required in the solution at time $k$ for the predicted step to be in the quadratic convergence region of the problem at time $k+1$.  
\end{proof}
%
The implication of the previous result is that the prediction--correction schemes yield a control law whose input is not very different from solving exactly the receding horizon problem \eqref{eqn_optimization_problem}. However, the main benefit is that the complexity of the algorithm proposed is of order $\log_2\log_2(1/\eta)$, where $\eta$ is the accuracy of the solution. {In particular, \textemdash as long as the variation of the system is not close to the limiting value assumed in Proposition \ref{prop_qcr} \textemdash the complexity of the algorithm is of order $O(\log_2(\log_2({B}T_s))$.} Conversely, we can establish a bound on the maximum variation allowed so that the total cost of the algorithm is that of two Hessian inversions. This is the subject of the following corollary
\begin{corollary}\label{coro}
  Under the hypotheses of Theorem \ref{theo_newton} in order to require at most one correction step we require the variation of the system to be bounded by 

  \begin{equation}
{B}T_s \leq \frac{m}{L}\frac{\sqrt{\delta_2^2+8\delta_1}-\delta_2}{4\delta_1}
  \end{equation}
  \end{corollary}
\begin{proof}
Notice that in order to be able to use only one prediction step, we require the bound in \eqref{eqn_number_corrections} to be less than one. This translates in the following condition 
  \begin{equation}
\frac{1+\delta_2{B}T_s\frac{L}{m}}{1-\frac{L^2}{m^2}\delta_1B^2T_s^2} \leq 2.
  \end{equation}
Which in turn can be rewritten as
  \begin{equation}
2\frac{L^2}{m^2}\delta_1 B^2T_s^2+\delta_2{B}T_s\frac{L}{m}-1 \leq 0.
  \end{equation}
  Finding the roots of the second order polynomial on $BT_s$ completes the proof of the result. 
\end{proof}
The previous result establishes a bound on the maximum allowable variation of the states of the system that makes the control possible with only one correction step, or equivalently with a computational cost that is equivalent to performing two Newton steps per control action computed.

The fact that the optimization problem \eqref{eqn_optimization_problem} can be solved with accuracy $\eta$, can be interpreted as if the control law designed is a perturbed version of MPC. Building on this fact -- and similar to \cite{paternain2018prediction}--  we can establish that the closed loop system \eqref{eqn_closed_loop} is Input to State Stable. This is the subject of the next section. 

%% file: stability.tex

\section{Input to state stability}\label{sec_stability}
We work next towards establishing input to state stability of the system \eqref{eqn_closed_loop} with respect to the approximation errors. In the previous section we showed that the error in solving the receding horizon problem \eqref{eqn_optimization_problem} by using the prediction-correction approach \eqref{eqn_prediction}--\eqref{eqn_correction} is bounded by $\eta m/L$ (cf., Theorem \ref{theo_newton}). Define then, the sequence of disturbances $\left\{\bbd(k), k \geq 0\right\}$ as
\begin{equation}
\bbd(k) = \bbu_{MPC}(k) - \bbu(k),
  \end{equation}
and write the dynamical system \eqref{eqn_closed_loop} as
\begin{equation}\label{eqn_closed_loop_new}
\bbx(k+1) = f_{\kappa}(\bbx(k)) = \tilde{f}_{\kappa_{MPC}}(\bbx(k), \bbd(k)),
\end{equation}
where $ \tilde{f}_{\kappa_{MPC}}(\bbx(k), \bbd(k))$ is such that
\begin{equation}\label{eqn_nominal_closed_loop}
 \tilde{f}_{\kappa_{MPC}}(\bbx(k), \bm{0}) = f_{\kappa_{MPC}}(\bbx(k)).
 \end{equation}
We define input to state stability of the closed loop system \eqref{eqn_closed_loop} where the input is the sequence of disturbances. These distrubances are the errors in solving \eqref{eqn_optimization_problem} using the prediction-correction scheme \eqref{eqn_prediction}--\eqref{eqn_correction}.
\begin{definition}[\bf Input to State Stability]\label{def_iss}
    Let $\tilde{g}:\mathbb{R}^n\times \mathbb{R}^w\to \mathbb{R}^n$ and consider the following discrete time dynamical system
    \begin{equation}\label{eqn_perturbed_dynamical_system}
      \bbx(k+1) = \tilde{g}(\bbx(k),\bbd(k)), 
      \end{equation}
with solution denoted by $\bm{\phi}_{\kappa}(k,\bbx(0))$. System \eqref{eqn_perturbed_dynamical_system} is input to state stable if there exists a $\ccalK\ccalL$-function $\beta$ and a $\ccalK$-function $\gamma$ such that for all initial state $\bbx(0)$ and sequence of disturbances $\left\{\bbd(k),k\geq0\right\}$ satisfying
  \begin{equation}\label{eqn_iss}
\left\|\bm{\phi}_{\kappa}(k,\bbx(0)) \right\|\leq \beta\left(\left\|\bbx(0) \right\|,k\right)+\gamma\left(\max_{j<k}\left\|\bbd(j)\right\|\right).
    \end{equation}
  We say that \eqref{eqn_perturbed_dynamical_system} is locally input to state stable if there exists constants $c_1$ and $c_2$ such that \eqref{eqn_iss} holds for any initial state $\left\|\bbx(0)\right\|\leq c_1$ and disturbances $\left\|\bbd(k)\right\|\leq c_2$. 
  \end{definition}
From Theorem \ref{theo_newton} one has that $\left\|\bbd(k) \right\|<\eta m/L$ for all $k \geq 0$. The bounds on the disturbances along with the following results in \cite{limon2009input} allows us to establish local input to state stability and input to state stability with stronger assumptions.
  \begin{theorem}[\bf Theorem 2 \cite{limon2009input}]\label{theo_iss}
 Let $\tilde{g}(\bbx,\bbd)$ in \eqref{eqn_perturbed_dynamical_system} be absolutely continuous in $\bbd$ for all $\bbx\in\mathbb{R}^n$ and for all $\bbd \in D\subset \mathbb{R}^w$, where $D$ is a compact set. If the nominal dynamical system ${g}(\bbx(k)) = \tilde{g}(\bbx(k),\bm{0})$ is asymptotically stable, the perturbed system \eqref{eqn_perturbed_dynamical_system} is Input to state stable if one of the following conditions holds
  \begin{itemize}
  \item Function ${g}(\bbx)$ is absolutely continuous in $\bbx\in\mathbb{R}^n$
  \item There exists a absolutely continuous Lyapunov function $V(\bbx)$ for the system $\bbx(k+1)={g}(\bbx(k))$.
  \end{itemize}
  \end{theorem}
\begin{corollary}[\bf Corollary 1\cite{limon2009input}]\label{coro_iss}
 Let $\tilde{g}(\bbx,\bbd)$ in \eqref{eqn_perturbed_dynamical_system} be continuous in a neighborhood of $\bbx=0$ and $\bbd=0$.  If the nominal dynamical system ${g}(\bbx(k)) = \tilde{g}(\bbx(k),\bm{0})$ is asymptotically stable, the perturbed system \eqref{eqn_perturbed_dynamical_system} is locally Input to state stable if one of the following conditions holds
      \begin{itemize}
  \item Function ${g}(\bbx)$ is continuous in a neighborhood of $\bbx = \bm{0}$
  \item There exists a Lyapunov function $V(\bbx)$ for the system $\bbx(k+1)={g}(\bbx(k))$ which is continuous in a neighborhood of $\bbx=\bm{0}$.
  \end{itemize}
\end{corollary}
Notice that the assumptions on the continuity of the closed loop system are not trivially satisfied by all systems since the continuity of the optimizer in \eqref{eqn_optimization_problem} is not generally  guaranteed for multi-parametric optimization problems. For instance some necessary conditions are linear equality constraints, continuous objective function and the solution of \eqref{eqn_optimization_problem} being unique and compact,  see e.g. \cite[Chapter 6]{borrelli2017predictive}. In the case here considered equality constraints are not linear, hence we need to assume these to establish input to state stability of the system \eqref{eqn_closed_loop} with respect to the disturbances. 
\begin{proposition}
  Let Assumptions \ref{assumption_continuity}--\ref{assumption_bounded_derivatives} hold. Then, if the system \eqref{eqn_closed_loop_new} is absolutely continuous in $\bbd$ for all $\bbx\in\mathbb{R}^n$ and for all $\bbd$ and \eqref{eqn_nominal_closed_loop} is absolutely continuous in $\bbx$ then \eqref{eqn_closed_loop_new} is input to state stable. If the system \eqref{eqn_closed_loop_new} is continuous in a neighborhood of $\bbx =0 $ and $\bbd =0$ and \eqref{eqn_nominal_closed_loop} is continuous in a neighborhood of $\bbx =0$ then \eqref{eqn_closed_loop_new} is locally input to state stable. 
\end{proposition}
\begin{proof}
Theorem \ref{theo_mpc_stability} establishes that the nominal system $f_{\kappa_{MPC}}(\bbx)$ is asymptotically stable. Then, the proof follows from the fact that the disturbances $\bbd(k)$ are bounded (Theorem \ref{theo_newton}) and the results from Theorem \ref{theo_iss} and Corollary \ref{coro_iss}.
  \end{proof}

%% file: numerical_examples.tex
\section{Numerical experiments}\label{sec_examples}
In this section we illustrate the theoretical results presented in the previous section by studying the problem of controlling the position of a point mass subject to nonlinear friction and the Hick's reactor \cite{hicks1971approximation}.
\subsection{Nonlinear Friction}\label{sec_friction}
\begin{figure}
  \centering
 \input{./figures/friction.tex}
  \caption{Nonlinear friction $F_a(\dot{x}) = 0.25\tanh(100\dot{x}))-0.25\tanh(10\dot{x})
    +0.1 \tanh(50\dot{x})+0.01\dot{x}$ considered in the numerical experiments.}
  \label{fig_nonlinear_friction}
    \end{figure}
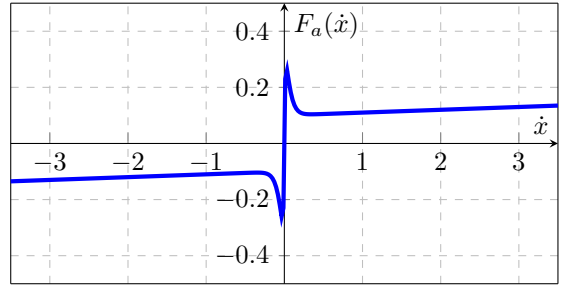
%
%
%
\begin{figure}
  \input{./figures/trajectory_ts200.tex}
  \caption{Control input, position and velocity of the system \eqref{eqn_dynamics_example} when using PC-MPC (Algorithm \ref{alg_pcmpc}) and N-MPC (Algorithm \ref{alg_mpc}) to solve \eqref{eqn_optimization_problem}. The sampling time is $T_s = 0.2 s$, the horizon is $H=5$ and we run both controllers for $200$ time steps. The accuracy for both algorithms is set to $\varepsilon = 10^{-2}$ and the maximum number of iterations per step is bounded by $N=50$. }\label{fig_trajectory200}
  \end{figure}
In this section we consider the problem of controlling the position of a point mass subject to nonlinear friction. Let us denote by $M$ the mass and let $x$ be its position, $g$ the gravitational constant and $F_a(\dot{x})$ the dynamic friction depending on the velocity of the mass and $u$ the control input force. With these definitions, consider the following dynamics 
\begin{equation}\label{eqn_dynamics_example}
M\ddot{x} =u-MgF_a(\dot{x}).
  \end{equation}
We discretize the dynamics using Euler's forward method to obtain a system of the form \eqref{eqn_dynamics}. To be precise, define the state vector $\bbx \in\mathbb{R}^2$ where its first and second components correspond to position and velocity respectively. Then, for sampling time $T_s$ the discretization yields
\begin{subequations}
\begin{equation}
  \bbx_1(k+1) = \bbx_1(k) +T_s \bbx_2(k)
\end{equation}
\begin{equation}
\bbx_2(k+1) =\bbx_2(k) +T_s\left(\frac{u(k)}{M} -g F_a(\bbx_2(k))\right).
\end{equation}
\end{subequations}
We consider a quadratic cost $\ell(\bbx,u) = \left(\bbx^\top\bbQ\bbx +u^2R\right)/2$, with $R>0$ and $\bbQ \in\ccalM^{2\times 2}$.
The specific values selected for the experiments are $g=9.81$, $M=0.2$, $\bbx(0)=[0.1, 0.1 ]^\top$, $\bbQ$ is diagonal with $\bbQ_{11} = 1\times 10^{3}$ and $\bbQ_{22} = 2$,
%
%
$R=1\times 10^{-3}$, $T_s = 0.2$, $H=5$, simulation time $T=40$ seconds and dynamic friction with the following expression 
\begin{equation}
  \begin{split}
    F_a(\dot{x})&= 0.25(\tanh(100\dot{x}))-\tanh(10\dot{x}))\\
    &+0.1 \tanh(50\dot{x})+0.01\dot{x}.
\end{split}
  \end{equation}
%
\begin{algorithm}
  \caption{\texttt{shift}}
  \label{alg_shift} 
\begin{algorithmic}[1]
 \renewcommand{\algorithmicrequire}{\textbf{Input:}}
 \renewcommand{\algorithmicensure}{\textbf{Output:}}
 \Require $\bbz\in\mathbb{R}^{H(2n_p)+2n}$
 \State Define a shift matrix for the states and the multipliers
 $$
 \bbC_{1} = \left[\begin{matrix} \boldsymbol{0}_H& \bbI_H \\ \boldsymbol{0}_H^\top& 1 \end{matrix}\right]
 $$
 \State Define a shift matrix for the inputs
  $$
 \bbC_{2} = \left[\begin{matrix} \boldsymbol{0}_{H-1}& \bbI_{H-1} \\ \boldsymbol{0}_{H-1}^\top& 1 \end{matrix}\right]
 $$
 \State Define the shift matrix for all variables
   $$
 \bbC = \left[\begin{matrix} \bbC_1\otimes\bbI_n & \boldsymbol{0}_{n(H+1)\times pH} & \boldsymbol{0}_{n(H+1)\times (H+1)} \\
     \boldsymbol{0}_{n(H+1)\times pH} & \bbC_2 \otimes\bbI_{p} &\boldsymbol{0}_{n(H+1)\times (H+1)} \\
   \boldsymbol{0}_{(H+1)\times pH} & \boldsymbol{0}_{(H+1)\times pH} & \bbC_1 \end{matrix}\right]
 $$
 \Return $\bbC\bbz$
 \end{algorithmic}
 \end{algorithm}
\begin{algorithm}
  \caption{Newton's method for MPC (N-MPC)}
  \label{alg_mpc} 
\begin{algorithmic}[1]
 \renewcommand{\algorithmicrequire}{\textbf{Input:}}
 \renewcommand{\algorithmicensure}{\textbf{Output:}}
 \Require $\bbx(0),N,\varepsilon$
 \State Compute $\bbz_0=\bbz^*(\bbx(0))$ 
 \For{$k=0,1,\ldots$}
 \State Apply input $\bbu(k) = (\bar{\bbu}_1)_k$ to the system and observe
 $$\bbx(k+1)=f(\bbx(k),\bbu(k))$$
 \State Re initialize seed $\bbz_{k+1}^0 = \texttt{shift}(\bbz_k)$
 \State Set $j=0$
 \While {$j<N$ or $\left\|\nabla_\bbz\ccalL(\bbx(k+1),\bbz_{k+1}^j)\right\|>\varepsilon$}
 \State $j=j+1$
 \State Compute Correction (or Newton) step as in \eqref{eqn_correction}
 $$
 \bbz_{k+1}^{j}= \bbz_{k+1}^{j-1}-\nabla^2_{\bbz\bbz} \ccalL_{j-1}(k+1)^{-1} \nabla_\bbz \ccalL_{j-1}(k+1),
 $$
 \EndWhile
 \State Update variable $\bbz_{k+1}=\bbz_{k+1}^j$
 \EndFor
%
 \end{algorithmic}
 \end{algorithm}
{We compare the behavior of the closed loop system when we use the Prediction-Correction Algorithm \eqref{eqn_prediction}-\eqref{eqn_correction} and when we use a Newton solver (Algorithm \ref{alg_mpc}). Observe that the only constraints that problem \eqref{eqn_optimization_problem} considers are the equality constraints that describe the evolution of the system. Thus, Algorithm \ref{alg_mpc} is equivalent to run a Sequential Quadratic Programming algorithm. Although not necessary, we add a shift to initialize the Newton solver (N-MPC) at the beginning of each control interval. This is done to propagate the solution computed at time $k$ to the next time step. Specifically, we set $(\bar{\bbu}_{i})_{k+1} = (\bar{\bbu}_{i+1})_k$  for all $i=1,\ldots, H-1$ and $(\bar{\bbx}_{i})_{k+1} = (\bar{\bbx}_{i+1})_k$, $(\bblambda_{i})_{k+1} = (\bblambda_{i+1})_k$ for all $i=1,\ldots, H$. This shift can be implemented using matrix multiplications as summarized under Algorithm \ref{alg_shift}.
  
  In Figure \ref{fig_trajectory200} we observe the control input, position and velocity resulting from solving the MPC problem \eqref{eqn_optimization_problem} using Algorithm \ref{alg_mpc}-- in blue-- and PC-MPC (Algorithm \ref{alg_pcmpc})-- in red. In this example we set the accuracy $\varepsilon =0.01$ and we bound the maximum number of Newton steps (correction steps) in both algorithms by $N=50$. 
  As it can be observed in Figure \ref{fig_trajectory200} the parameters selected are not sufficient to achieve successful control since the inputs computed by Algorithm \ref{alg_mpc} -- depicted in blue -- fail to drive the system to the origin. It is important to point out that in all control intervals the solutions satisfy the required accuracy $\varepsilon = 0.01$ and that none of the computations of the control actions involves more than $25$ Newton steps. Hence we conclude that the accuracy required is not enough to achieve a good performance of the system.  
%
  %
  %
\begin{figure*}
  \begin{tikzpicture}
    \node[anchor=south west] (img) at (0,0) {\input{figures/error_ts200acc.tex}};
  \node[anchor=south west] (img2) at (0.5\linewidth,0) {\input{figures/iterations_ts200acc.tex}};  \end{tikzpicture}
  ~ \caption{Histograms with frequency of error in the solution and frequency of the number of iterations required by PC-MPC (Algorithm \ref{alg_pcmpc}) and N-MPC (Algorithm \ref{alg_mpc}) to solve \eqref{eqn_optimization_problem}. The sampling time is $T_s = 0.2 s$, the horizon is $H=5$ and we run both controllers for $200$ time steps. The accuracy for both algorithms is set to $\varepsilon = 10^{-4}$ and the maximum number of iterations per step is bounded by $N=50$. }\label{fig_hist200acc}
\end{figure*}
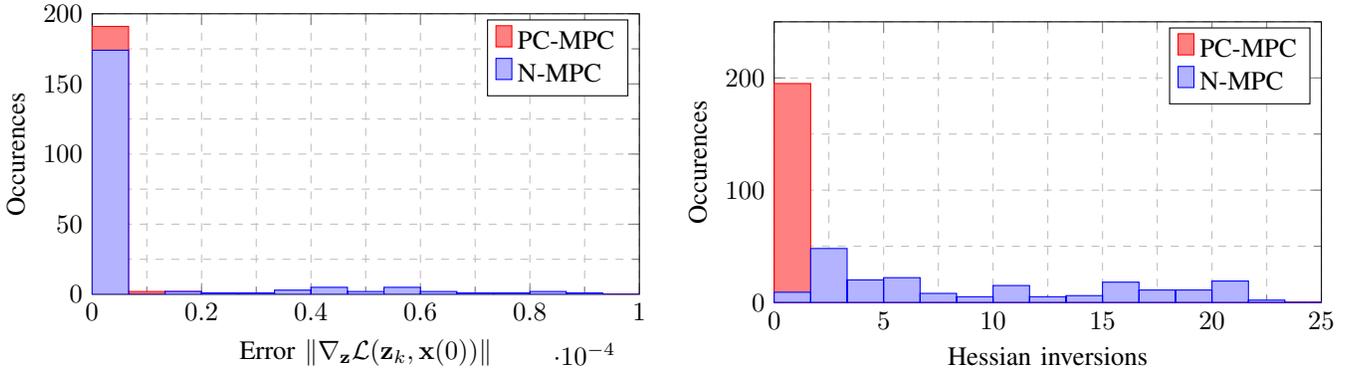
%
\begin{figure}
  \input{./figures/trajectory_ts200acc.tex}
  \caption{Control input, position and velocity of the system \eqref{eqn_dynamics_example} when using PC-MPC (Algorithm \ref{alg_pcmpc}) and N-MPC (Algorithm \ref{alg_mpc}) to solve \eqref{eqn_optimization_problem}. The sampling time is $T_s = 0.2 s$, the horizon is $H=5$ and we run both controllers for $200$ time steps. The accuracy for both algorithms is set to $\varepsilon = 10^{-4}$ and the maximum number of iterations per step is bounded by $N=50$. }\label{fig_ts200acc}
  \end{figure}
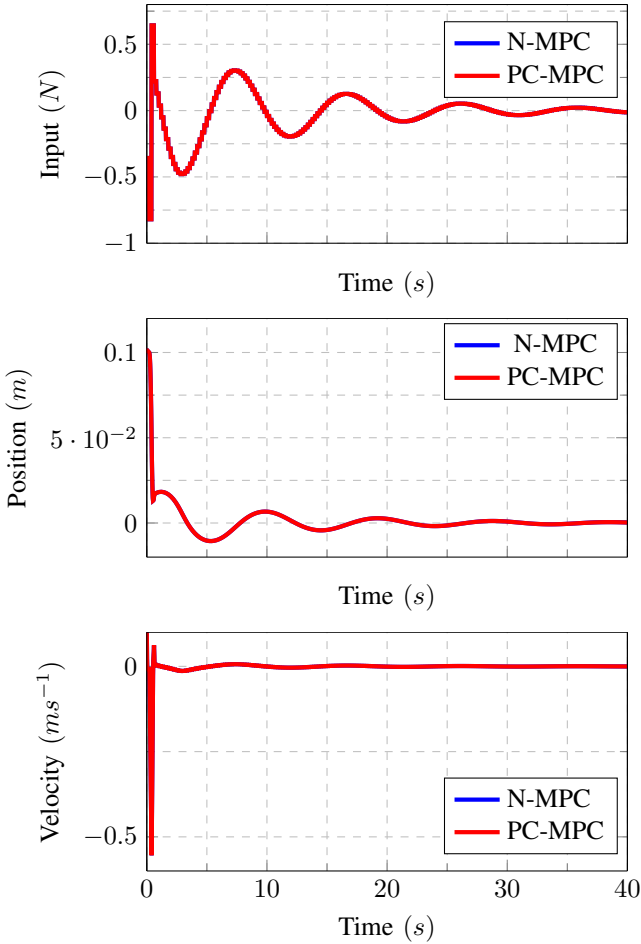
  Based on this observation, we modify the desired accuracy of the solution to $\varepsilon =0.0001$. As it can be observed in Figure \ref{fig_ts200acc} the two trajectories are now indistinguishable but from a computational point of view PC-MPC outperforms the classic Newton algorithm (Algorithm \ref{alg_mpc}). This claim is supported by the histogram in Figure \ref{fig_hist200acc} where we can observe that the number of Hessian inversions required to achieve the desired accuracy in the case of PC-MPC is considerably smaller. Specifically, the Newton method (Algorithm \ref{alg_mpc}) requires on average $9$ iterations while PC-MPC requires only one in 195 of the 200 time steps and two iterations in the remaining five time steps. }

\subsection{Hicks Reactor}

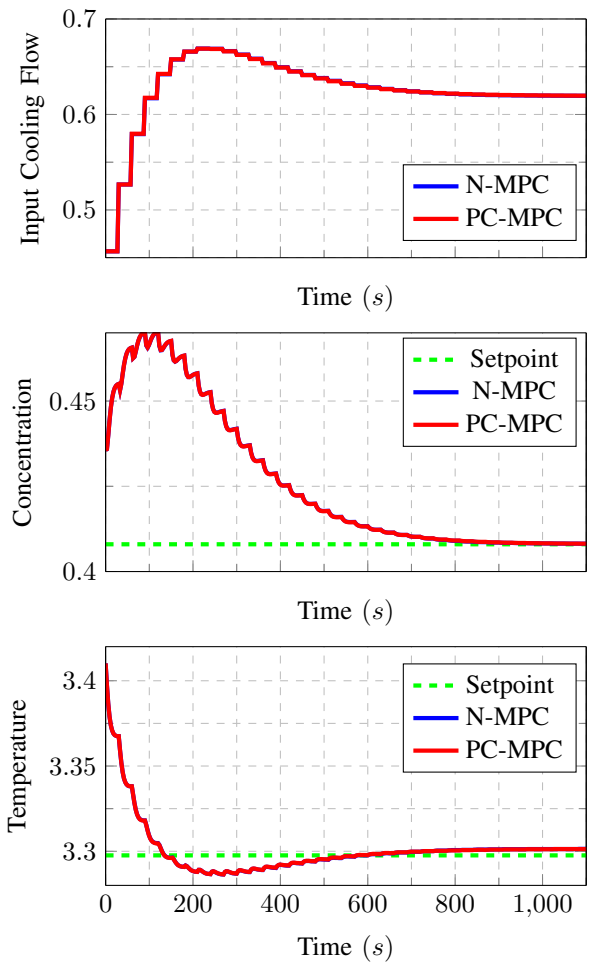
\begin{figure}
  \input{figures/trajectory_hicks.tex}
                ~ \caption{Control input, concentration and temperature of the Hicks reactor \eqref{eqn_hicks_reactor} when using PC-MPC (Algorithm \ref{alg_pcmpc}) and N-MPC (Algorithm \ref{alg_mpc}) to solve \eqref{eqn_optimization_problem}. The sampling time is $T_s = 30 s$, the horizon is $H=10$ and we run both controllers for $40$ time steps. The accuracy for both algorithms is set to $\varepsilon = 10^{-3}$ and the maximum number of iterations per step is bounded by $N=100$. As it can be observed there is not a difference in the trajectories of both methods which means that the solution of the receding horizon problem that is being computed is the same.}\label{fig_plots_hicks}
\end{figure}
%
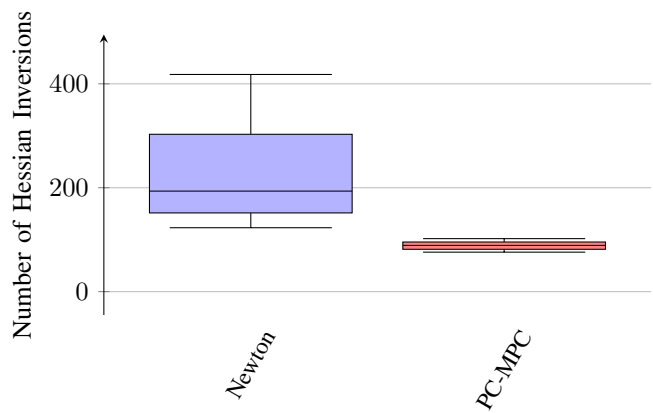
\begin{figure}
  \input{figures/boxplot.tex}
  \caption{Box-plots with the total number of Hessian inversions required to control the Hicks reactor \eqref{eqn_hicks_reactor}  when using PC-MPC (Algorithm \ref{alg_pcmpc}) and N-MPC (Algorithm \ref{alg_mpc}) to solve \eqref{eqn_optimization_problem}. The sampling time is $T_s = 30 s$, the horizon is $H=10$ and we run both controllers for $40$ time steps. The accuracy for both algorithms is set to $\varepsilon = 10^{-3}$ and the maximum number of iterations per step is bounded by $N=100$. As it can be observed, in term of computational computational complexity there is a benefit in the computational complexity of Prediction-Correction MPC since it reduces the number of iterations required. In particular, on average Prediction-Correction requires $43\%$ of the iterations requires by MPC.}
                \label{fig_boxplot}
  \end{figure}

In this subsection we consider the CSTR system \cite{hicks1971approximation}, which has been used as a benchmark in the control literature \cite{huang2012robust, griffith2018advances}. In this problem we are interested in controlling the concentration $z_c$ and the temperature $z_T$ of a chemical process. The control input available is the cooling water flow $u$. The dimensionless temperature and concentration are subject to the following dynamical system
\begin{equation}\label{eqn_hicks_reactor}
  \begin{split}
&    \dot{z}_c = \frac{1-z_c}{\theta}-k_0 z_ke^{-Ea/z_T}\\
 &   \dot{z}_T = \frac{z_T^f-z_T}{\theta}+k_0 z_ke^{-Ea/z_T}-\nu U_{1sf}u_1(z_T-z_T^{cw})
    \end{split}
  \end{equation}
where $z_T^{cw}=2.9, z_{T}^f = 3, Ea=25.2, \nu=1.95\times 10^{-4}, k_0 = 300, U_{1sf}=600$ and $\theta=10$. Define the state vector $\bbx= [z_c ,z_T]^\top$ then the goal is to stabilize the system at the steady state $\bbx_{ss} = [0.408,3.29763]^\top$ with steady state control input is $u_{ss} = 0.6167$ \cite{hicks1971approximation}. We propose to control the system by minimizing the following cost
\begin{equation}
\ell(\bbx,\bbu) = (\bbx-\bbx_{ss})^\top\bbQ(\bbx-\bbx_{ss})+|u-u_{ss}|^2,
  \end{equation}
where $\bbQ=\diag([10,2])$ with control intervals of $T_s=30$ seconds. We implement an Euler discretization of the dynamics \eqref{eqn_hicks_reactor} to write an optimization problem of the form in \eqref{eqn_optimization_problem}. We run 100 simulations for $1200$ seconds using both the Newton method (Algorithm \ref{alg_mpc}) and the Prediction-Correction method proposed (Algorithm \ref{alg_pcmpc}) for different initial conditions selected as $\bbx_0 = \bbx_{ss}(1+0.2\xi)$, where $\xi\sim\ccalN(0,1)$. In Figure \ref{fig_plots_hicks} we plot the control input, concentration and temperature of one of the simulations using both algorithms. As it can be observed there is no difference in the trajectories of both methods which means that the solution to the receding horizon problem that is being computed is the same. However, there is a significant gain in terms of the computation cost as it can be observed in Figure \ref{fig_boxplot} were we present a box-plot with the number of iterations required per simulation using the two different schemes. Notice that in particular the maximum number of Hessian inversions required by Prediction-Correction MPC is smaller than the minimum number of Hessian inversions required by the Newton method. Likewise the average reduction in computation is around $57\%$.

%% file: figures/friction.tex
\begin{tikzpicture}
		\pgfplotsset{grid style={dashed,gray!50}}
		\begin{axis}[
                       clip =false,
			xmin=-3.5,xmax=3.5,ymin=-0.5,ymax=0.5,
                        axis lines=center,
                        axis on top =false,
			xlabel=$\dot{x}$,
			ylabel=$F_a(\dot{x})$,
                        xtick align=inside,
			width=\columnwidth, 
			height=0.6\columnwidth,
			grid=both,
		  ]
                  \draw (rel axis cs:0,0) -- (rel axis cs:1,0)
                  (rel axis cs:0,1) -- (rel axis cs:1,1);
                  \draw (rel axis cs:0,0) -- (rel axis cs:0,1);
                  \draw (rel axis cs:1,0) -- (rel axis cs:1,1);
               	    \addplot[domain=-3.5:3.5,samples=300,mark=none,draw=blue,ultra thick] {0.25*tanh(100*\x))-0.25*tanh(10*\x)+0.1*tanh(50*\x)+0.01*\x};
	\end{axis}
	\end{tikzpicture}

%% file: figures/trajectory_ts200.tex
\begin{tikzpicture}
    \pgfplotsset{grid style={dashed,gray!50}}
    \begin{groupplot}[
    xlabel near ticks,
    ylabel near ticks,
    width=0.9\columnwidth,
    xmin=0,xmax=40,ymin=-0.6,ymax=0.6,
        xlabel= Time $(s)$,
    group style={
        group size=1 by 3,
        xticklabels at=edge bottom,
        yticklabels at=edge left,
    },
    grid=both,
    minor tick num=1,
    axis x line*=left,
    ]
    
   \nextgroupplot[y post scale=0.6,xmin=0,xmax=40,ymin=-1, ymax=0.8,ylabel=Input $(N)$, legend columns=1,
            legend style={
                at={(0.98,0.93)},
                anchor=north east,
                legend cell align=left
                } ]
            
  \addplot [color=blue,solid,ultra thick,mark=none]  table[x index=0,y index=1, col sep=comma]{data/friction/ts_200/ts200.dat};	\addlegendentry{N-MPC};
        \addplot [color=red,solid,ultra thick,mark=none]
	table[x index=0,y index=4, col sep=comma]{data/friction/ts_200/ts200.dat};	\addlegendentry{PC-MPC};

\draw  (rel axis cs:0,1) -- (rel axis cs:1,1);
                  \draw (rel axis cs:0,0) -- (rel axis cs:0,1);
                  \draw (rel axis cs:1,0) -- (rel axis cs:1,1);

    \nextgroupplot[y post scale=0.6, ymin=-0.02, ymax=0.12, try min ticks=2,ylabel=Position $(m)$]
legend columns=2,
            legend style={
                at={(0.98,0.93)},
                anchor=north east,
                legend cell align=left
                }  
     \addplot [color=blue,solid,ultra thick,mark=none]
	table[x index=0,y index=2, col sep=comma]{data/friction/ts_200/ts200.dat};
	\addlegendentry{N-MPC};
        \addplot [color=red,solid,ultra thick,mark=none]
	table[x index=0,y index=5, col sep=comma]{data/friction/ts_200/ts200.dat};
	\addlegendentry{PC-MPC};
        \draw  (rel axis cs:0,1) -- (rel axis cs:1,1);
                  \draw (rel axis cs:0,0) -- (rel axis cs:0,1);
                  \draw (rel axis cs:1,0) -- (rel axis cs:1,1);

    \nextgroupplot[y post scale=0.6,ymin=-0.6, ymax=0.1, try min ticks=2,ylabel=Velocity $(ms^{-1})$,legend columns=1,
            legend style={
                at={(0.98,0.07)},
                anchor=south east,
                legend cell align=left
                } ]

	\addplot [color=blue,solid,ultra thick,mark=none] table[x index=0,y index=3, col sep=comma]{data/friction/ts_200/ts200.dat};
	\addlegendentry{N-MPC};
        \addplot [color=red,solid,ultra thick,mark=none]
	table[x index=0,y index=6, col sep=comma]{data/friction/ts_200/ts200.dat};
	\addlegendentry{PC-MPC};
        \draw  (rel axis cs:0,1) -- (rel axis cs:1,1);
                  \draw (rel axis cs:0,0) -- (rel axis cs:0,1);
                  \draw (rel axis cs:1,0) -- (rel axis cs:1,1);

    \end{groupplot}
\end{tikzpicture}

%% file: figures/error_ts200acc.tex
\begin{tikzpicture}
  \pgfplotsset{grid style={dashed,gray!50}}
		\begin{axis}[
			ybar,	
			minor tick num=1,
			xtick align=inside,
			xlabel={Error $\left\|\nabla_\bbz\ccalL(\bbz_k,\bbx(0))\right\|$},
			ylabel={Occurences},
			ylabel near ticks,
			width=\columnwidth, 
			height=0.6\columnwidth,
			xmin=0,xmax=0.0001,ymin=0,ymax=200,
			grid=both,
			legend style={
				at={(0.98,0.98)}, 
				anchor=north east, 
				legend cell align=left
				}
		]     
               	\addplot[red, fill=red!50, hist={data=x, bins=15}] file {data/friction/ts_200/error_tvmpc_tolacc.dat};		
\addlegendentry{PC-MPC};

	\addplot[blue, fill=blue!30, hist={data=x, bins=15}] file {data/friction/ts_200/error_mpc_tolacc.dat};              	
		\addlegendentry{N-MPC};
	\end{axis}

\end{tikzpicture}

%% file: figures/iterations_ts200acc.tex
\begin{tikzpicture}
		\pgfplotsset{grid style={dashed,gray!50}}
		\begin{axis}[
			ybar,	
			minor tick num=1,
			xtick align=inside,
			xlabel={Hessian inversions},
			ylabel={Occurences},
			ylabel near ticks,
			width=\columnwidth, 
			height=0.6\columnwidth,
			xmin=0,xmax=25,ymin=0,ymax=250,
			grid=both,
			legend style={
				at={(0.98,0.98)}, 
				anchor=north east, 
				legend cell align=left
				}
		]     
               	\addplot[red, fill=red!50, hist={data=x, bins=15}] file {data/friction/ts_200/iter_tvmpc_tolacc.dat};		
\addlegendentry{PC-MPC};

	\addplot[blue, fill=blue!30, hist={data=x, bins=15}] file {data/friction/ts_200/iter_mpc_tolacc.dat};              	
		\addlegendentry{N-MPC};
	\end{axis}
	\end{tikzpicture}

%% file: figures/trajectory_ts200acc.tex
\begin{tikzpicture}
    \pgfplotsset{grid style={dashed,gray!50}}
    \begin{groupplot}[
    xlabel near ticks,
    ylabel near ticks,
    width=0.9\columnwidth,
    xmin=0,xmax=40,ymin=-0.6,ymax=0.6,
        xlabel= Time $(s)$,
    group style={
        group size=1 by 3,
        xticklabels at=edge bottom,
        yticklabels at=edge left,
    },
    grid=both,
    minor tick num=1,
    axis x line*=left,
    ]
    
   \nextgroupplot[y post scale=0.6,xmin=0,xmax=40,ymin=-1, ymax=0.8,ylabel=Input $(N)$, legend columns=1,
            legend style={
                at={(0.98,0.93)},
                anchor=north east,
                legend cell align=left
                } ]
            
  \addplot [color=blue,solid,ultra thick,mark=none]  table[x index=0,y index=1, col sep=comma]{data/friction/ts_200/ts200acc.dat};	\addlegendentry{N-MPC};
        \addplot [color=red,solid,ultra thick,mark=none]
	table[x index=0,y index=4, col sep=comma]{data/friction/ts_200/ts200acc.dat};	\addlegendentry{PC-MPC};

\draw  (rel axis cs:0,1) -- (rel axis cs:1,1);
                  \draw (rel axis cs:0,0) -- (rel axis cs:0,1);
                  \draw (rel axis cs:1,0) -- (rel axis cs:1,1);

    \nextgroupplot[y post scale=0.6, ymin=-0.02, ymax=0.12, try min ticks=2,ylabel=Position $(m)$]
legend columns=2,
            legend style={
                at={(0.98,0.93)},
                anchor=north east,
                legend cell align=left
                }  
     \addplot [color=blue,solid,ultra thick,mark=none]
	table[x index=0,y index=2, col sep=comma]{data/friction/ts_200/ts200acc.dat};
	\addlegendentry{N-MPC};
        \addplot [color=red,solid,ultra thick,mark=none]
	table[x index=0,y index=5, col sep=comma]{data/friction/ts_200/ts200acc.dat};
	\addlegendentry{PC-MPC};
        \draw  (rel axis cs:0,1) -- (rel axis cs:1,1);
                  \draw (rel axis cs:0,0) -- (rel axis cs:0,1);
                  \draw (rel axis cs:1,0) -- (rel axis cs:1,1);

    \nextgroupplot[y post scale=0.6,ymin=-0.6, ymax=0.1, try min ticks=2,ylabel=Velocity $(ms^{-1})$,legend columns=1,
            legend style={
                at={(0.98,0.07)},
                anchor=south east,
                legend cell align=left
                } ]

	\addplot [color=blue,solid,ultra thick,mark=none] table[x index=0,y index=3, col sep=comma]{data/friction/ts_200/ts200acc.dat};
	\addlegendentry{N-MPC};
        \addplot [color=red,solid,ultra thick,mark=none]
	table[x index=0,y index=6, col sep=comma]{data/friction/ts_200/ts200acc.dat};
	\addlegendentry{PC-MPC};
        \draw  (rel axis cs:0,1) -- (rel axis cs:1,1);
                  \draw (rel axis cs:0,0) -- (rel axis cs:0,1);
                  \draw (rel axis cs:1,0) -- (rel axis cs:1,1);

    \end{groupplot}
\end{tikzpicture}

%% file: figures/trajectory_hicks.tex
\begin{tikzpicture}
    \pgfplotsset{grid style={dashed,gray!50}}
    \begin{groupplot}[
    xlabel near ticks,
    ylabel near ticks,
    width=0.9\columnwidth,
    xmin=0,xmax=1100,ymin=0.6,ymax=3.42,
    xlabel= Time $(s)$,
    group style={
        group size=1 by 3,
        xticklabels at=edge bottom,
        yticklabels at=edge left,
    },
    grid=both,
    minor tick num=1,
    axis x line*=left,
    ]
    
   \nextgroupplot[y post scale=0.6,ymin=0.45, ymax=0.7,ylabel=Input Cooling Flow, legend columns=1,
            legend style={
                at={(0.98,0.07)},
                anchor=south east,
                legend cell align=left
                } ]
            
  \addplot [color=blue,solid,ultra thick,mark=none]  table[x index=0,y index=1, col sep=comma]{data/hicks/trajectory.dat};	\addlegendentry{N-MPC};
        \addplot [color=red,solid,ultra thick,mark=none]
	table[x index=0,y index=4, col sep=comma]{data/hicks/trajectory.dat};	\addlegendentry{PC-MPC};

\draw  (rel axis cs:0,1) -- (rel axis cs:1,1);
                  \draw (rel axis cs:0,0) -- (rel axis cs:0,1);
                  \draw (rel axis cs:1,0) -- (rel axis cs:1,1);

    \nextgroupplot[y post scale=0.6, ymin=0.4, ymax=0.47, try min ticks=2,ylabel=Concentration]
legend columns=2,
            legend style={
                at={(0.98,0.93)},
                anchor=north east,
                legend cell align=left
            }
            \addplot [domain=0:1100, samples=100, draw=green, dashed, ultra thick]{0.408};\addlegendentry{Setpoint};
     \addplot [color=blue,solid,ultra thick,mark=none]
	table[x index=0,y index=2, col sep=comma]{data/hicks/trajectory.dat};
	\addlegendentry{N-MPC};
        \addplot [color=red,solid,ultra thick,mark=none]
	table[x index=0,y index=5, col sep=comma]{data/hicks/trajectory.dat};
	\addlegendentry{PC-MPC};
        \draw  (rel axis cs:0,1) -- (rel axis cs:1,1);
        \draw (rel axis cs:0,0) -- (rel axis cs:0,1);
                  \draw (rel axis cs:1,0) -- (rel axis cs:1,1);

    \nextgroupplot[y post scale=0.6,ymin=3.28, ymax=3.42, try min ticks=2,ylabel=Temperature ,legend columns=1,
            legend style={
                at={(0.98,0.93)},
                anchor=north east,
                legend cell align=left
    } ]
    \addplot [domain=0:1100, samples=100, draw=green, dashed, ultra thick]{3.2976};\addlegendentry{Setpoint};
	\addplot [color=blue,solid,ultra thick,mark=none] table[x index=0,y index=3, col sep=comma]{data/hicks/trajectory.dat};
	\addlegendentry{N-MPC};
        \addplot [color=red,solid,ultra thick,mark=none]
	table[x index=0,y index=6, col sep=comma]{data/hicks/trajectory.dat};
	\addlegendentry{PC-MPC};
        \draw  (rel axis cs:0,1) -- (rel axis cs:1,1);
                  \draw (rel axis cs:0,0) -- (rel axis cs:0,1);
                  \draw (rel axis cs:1,0) -- (rel axis cs:1,1);
    \end{groupplot}
\end{tikzpicture}

%% file: figures/boxplot.tex
\begin{tikzpicture}
	\begin{axis}[
		boxplot/draw direction = y,
		x axis line style = {opacity=0},
		axis x line* = bottom,
		axis y line = left,
		enlarge y limits,
		ymajorgrids,
		xtick = {1,2},
		xticklabel style = {align=center, font=\small, rotate=60},
		xticklabels = {Newton, PC-MPC},
		xtick style = {draw=none}, 
		ylabel = {Number of Hessian Inversions},
                ymin=0,
                ymax=450,
                width=\columnwidth, 
			height=0.6\columnwidth,
	]
          \addplot+[boxplot prepared={lower whisker=123,lower quartile=151.5, median=193.5,upper quartile=302.875, upper whisker=418,
    }, fill=blue!30,draw=black,
  ] coordinates{};
	
          \addplot+[boxplot prepared={lower whisker=76,lower quartile=81.5, median=89,upper quartile=95.5,
              upper whisker=102,
    },fill=red!50,draw=black,
          ] coordinates{};

	\end{axis}
\end{tikzpicture}

%% file: conclusion.tex

\section{Conclusion}\label{sec_conclusion}
  We considered a prediction-correction algorithm to solve approximately a time varying multiparameter optimization problem in the context of Model Predictive Control. In particular the algorithm is such that the prediction step guarantees that the seed for the correction steps is in the quadratic convergence region as long as the variation of the states of the system is not too fast and the solution in the previous step has been computed with sufficient accuracy. Moreover, we show that the cost of achieving the sufficient accuracy is of order $\log_2\log_2(BT_s)$, where $BT_s$ is a bound on the variation of the state in a control interval. Based on these bounds, we further establish that the closed loop system with the control defined by the solution of the proposed algorithm is input to state stable with respect to the approximation error. We illustrated the computational advantages of the proposed method as compared to regular MPC in the example of the control of a point mass under a non-linear dynamic friction and in the Hicks reactor.

%% file: appendix.tex

\appendix

\subsection{Auxiliary Lemmas}
\begin{lemma}\label{lemma_aux1}
Let Assumptions \ref{assumption_non_zero_eigenvalues} and \ref{assumption_bounded_derivatives} hold and let $m$ and $L$ be the constants defined in said assumptions. Likewise, let $\delta_1$ be the constant defined in \eqref{eqn_delta1}. Under these conditions, it follows that $\nabla_{\bbx}\bbz^\star(\bbx)$ is Lipschitz with constant $2L\delta_1/m$. This is, for any $\bbx_1,\bbx_2\in\mathbb{R}^n$ it holds that
\begin{equation}
\left\|\nabla_{\bbx}\bbz^\star(\bbx_1) -\nabla_{\bbx}\bbz^\star(\bbx_2)\right\|\leq 2\frac{L}{m}\delta_1\left\|\bbx_1-\bbx_2\right\|.
\end{equation}
\end{lemma}
\begin{proof}
Recall that the local solutions of the optimization problem \eqref{eqn_optimization_problem} satisfy
\begin{equation}
\nabla_{\bbz}\ccalL(\bbx,\bbz^\star(\bbx)) = 0. 
\end{equation}
Assumption \ref{assumption_non_zero_eigenvalues} and the Implicit Function Theorem guarantee that we can write the derivative of $\bbz^\star(\bbx)$ as
\begin{equation}\label{eqn_derivative_zstar}
\nabla_{\bbx}\bbz^\star(\bbx) = - \nabla^2_{\bbz\bbz}\ccalL\left(\bbx,\bbz^\star(\bbx)\right)^{-1}\nabla^2_{\bbz\bbx}\ccalL\left(\bbx,\bbz^\star(\bbx)\right).
\end{equation}
Then write the difference of $\nabla_{\bbx}\bbz^\star(\bbx)$ evaluated at two different states $\bbx_1,\bbx_2 \in\mathbb{R}^n$ as
\begin{equation}
\begin{split}
\nabla_{\bbx}\bbz^\star(\bbx_2)-\nabla_{\bbx}\bbz^\star(\bbx_1) &=\nabla_{\bbz\bbz}^2\ccalL_1^{-1}\nabla_{\bbz\bbx}^2\ccalL_1 \\
&-\nabla_{\bbz\bbz}^2\ccalL_2^{-1}\nabla_{\bbz\bbx}^2\ccalL_2,
\end{split}
\end{equation}
where we defined for simplicity $\nabla_{\bbz\bbz}^2\ccalL_i:=\nabla_{\bbz\bbz}^2\ccalL(\bbx_i,\bbz^\star(\bbx_i))$ and $\nabla_{\bbz\bbx}^2\ccalL_i:=\nabla_{\bbz\bbx}^2\ccalL(\bbx_i,\bbz^\star(\bbx_i))$. Adding and subtracting $\nabla_{\bbz\bbz}^2\ccalL_1^{-1}\nabla_{\bbz\bbx}^2\ccalL_2$ to the previous expression yields
\begin{equation}
\begin{split}
\nabla_{\bbx}\bbz^\star(\bbx_2)-\nabla_{\bbx}\bbz^\star(\bbx_1) &=\nabla_{\bbz\bbz}^2\ccalL_1^{-1}\left(\nabla_{\bbz\bbx}^2\ccalL_1-\nabla_{\bbz\bbx}^2\ccalL_2\right) \\
&\left(\nabla_{\bbz\bbz}^2\ccalL_1^{-1}-\nabla_{\bbz\bbz}^2\ccalL_2^{-1}\right)\nabla_{\bbz\bbx}^2\ccalL_2.
\end{split}
\end{equation}
Using the triangle inequality, the Cauchy-Schwartz inequality and the bounds of Assumption \ref{assumption_non_zero_eigenvalues} and Assumption \ref{assumption_bounded_derivatives} it follows that the norm of the previous expression can be upper bounded by 
\begin{equation}\label{new_lemma_eq1}
\begin{split}
\left\|\nabla_{\bbx}\bbz^\star(\bbx_2)-\nabla_{\bbx}\bbz^\star(\bbx_1)\right\| &\leq \frac{1}{2m}\left\|\nabla_{\bbz\bbx}^2\ccalL_1-\nabla_{\bbz\bbx}^2\ccalL_2\right\| \\
&+C\left\|\nabla_{\bbz\bbz}^2\ccalL_1^{-1}-\nabla_{\bbz\bbz}^2\ccalL_2^{-1}\right\|.
\end{split}
\end{equation}
Notice that the second norm can be written as
\begin{equation}
\begin{split}
&\left\|\nabla_{\bbz\bbz}^2\ccalL_1^{-1}-\nabla_{\bbz\bbz}^2\ccalL_2^{-1}\right\| = \\
&\left\|\nabla_{\bbz\bbz}^2\ccalL_1^{-1}\left(\nabla_{\bbz\bbz}^2\ccalL_2-\nabla_{\bbz\bbz}^2\ccalL_1\right)\nabla_{\bbz\bbz}^2\ccalL_2^{-1}\right\|.
\end{split}
\end{equation}
Thus using the bound on the minimum eigenvalue of $\nabla^2_{\bbz\bbz}\ccalL(\bbx,\bbz^\star(\bbx)$ in Assumption \ref{assumption_non_zero_eigenvalues} it follows that 
\begin{equation}\label{new_lemma_eq2}
\begin{split}\left\|\nabla_{\bbz\bbz}^2\ccalL_1^{-1}-\nabla_{\bbz\bbz}^2\ccalL_2^{-1}\right\| \leq \frac{1}{4m^2}\left\|\nabla_{\bbz\bbz}^2\ccalL_2-\nabla_{\bbz\bbz}^2\ccalL_1\right\|.
\end{split}
\end{equation}
Combining the bounds established in \eqref{new_lemma_eq1} and \eqref{new_lemma_eq2} it follows that 
\begin{equation}
\begin{split}
\left\|\nabla_{\bbx}\bbz^\star(\bbx_2)-\nabla_{\bbx}\bbz^\star(\bbx_1)\right\| &\leq \frac{1}{2m}\left\|\nabla_{\bbz\bbx}^2\ccalL_1-\nabla_{\bbz\bbx}^2\ccalL_2\right\| \\
&+\frac{C}{4m^2}\left\|\nabla_{\bbz\bbz}^2\ccalL_1-\nabla_{\bbz\bbz}^2\ccalL_2\right\|.
\end{split}
\end{equation}
Since the derivatives of $\nabla_{\bbz}\ccalL(\bbx,\bbz)$ are Lipschitz (cf., Assumption \ref{assumption_bounded_derivatives}) the previous bound reduces to 
\begin{equation}
\begin{split}
\left\|\nabla_{\bbx}\bbz^\star(\bbx_2)-\nabla_{\bbx}\bbz^\star(\bbx_1)\right\| \leq \frac{L}{2m}\left(1+\frac{C}{2m}\right)\left\|\bby_1-\bby_2\right\|, 
\end{split}
\end{equation}
where we have used the notation $\bby_i:= (\bbx_i,\bbz^\star(\bbx_i))$. To complete the proof it suffices to show that $\left\|\bby_1-\bby_2\right\| \leq \left(1+C^2/(4m^2)\right)^{1/2}\left\|\bbx_1-\bbx_2\right\|$. We turn our focus into showing the latter. Notice that $\left\|\bby_1-\bby_2\right\|$ can be written as
\begin{equation}
\left\|\bby_1-\bby_2\right\| = \left(\left\|\bbx_1-\bbx_2\right\|^2+\left\|\bbz^\star(\bbx_1)-\bbz^\star(\bbx_2)\right\|^2\right)^{1/2}. 
\end{equation}
Thus, it suffices to show that
\begin{equation}
\left\|\bbz^\star(\bbx_1)-\bbz^\star(\bbx_2)\right\| \leq \frac{C}{2m}\left\|\bbx_1-\bbx_2\right\|.
\end{equation}
To see why this is the case write
\begin{equation}\label{eqn_integral_star}
\bbz^\star(\bbx_2)-\bbz^\star(\bbx_1) = \int_0^1 \nabla_{\bbx}\bbz^\star(\bbx_1+\theta \Delta \bbx)\Delta \bbx d \theta,
\end{equation}
with $\Delta\bbx = \bbx_2-\bbx_1$. Using the expression arising from the Implicit Function Theorem (cf., \eqref{eqn_derivative_zstar}) and considering Assumption \ref{assumption_non_zero_eigenvalues} and Assumption \ref{assumption_bounded_derivatives}, the norm of $\nabla_\bbz \bbz^\star(\bbx)$ can be bounded by
\begin{equation}
\left\|\nabla_{\bbx}\bbz^\star(\bbx)\right\|\leq \frac{C}{2m}. 
\end{equation}
Taking the norm in \eqref{eqn_integral_star} and substituting the previous bound in it reduces to
\begin{equation}
\begin{split}
\left\|\bbz^\star(\bbx_2)-\bbz^\star(\bbx_1)\right\| &\leq \int_0^1 \left\|\nabla_{\bbx}\bbz^\star(\bbx_1+\theta \Delta \bbx)\right\| \left\|\Delta \bbx \right\|d \theta \\
&\leq\int_0^1 \frac{C}{2m}\left\|\Delta \bbx \right\|d \theta=\frac{C}{2m}\left\|\Delta \bbx \right\|.
\end{split}
\end{equation}
The latter completes the proof of the result. 
\end{proof}

\begin{lemma}\label{lemma_bound_gradient}
Let $\ccalL(\bbx,\bbz)$ satisfy Assumptions \ref{assumption_non_zero_eigenvalues} and \ref{assumption_bounded_derivatives}, then for any $\bbz$ such that $\left\|\bbz-\bbz(\bbx) \right\|<m/L$ the following lower bound for the norm of the gradient of $\ccalL(\cdot,\cdot)$ with respect to its second argument holds
\begin{equation}
m \left\|\bbz-\bbz^\star(\bbx) \right\| < \left\|\nabla_z\ccalL(\bbx,\bbz) \right\|.
\end{equation}
In addition it holds that
\begin{equation}
\min_{i=1\ldots (n+p+1)H} \left| \lambda_i\left(\nabla^2_{\bbz \bbz}\ccalL(\bbx,\bbz)\right)\right|>m.
\end{equation}
\end{lemma}
\begin{proof}
The result follows from \cite[Corollary 1.2.1]{nesterov2013introductory} and the Lipschitz continuity of the second derivatives of the gradient. See for instance \cite[Appendix A.2]{paternain2019newton} for a proof.
\end{proof}


\subsection{Proof of Proposition \ref{prop_optimum_prediction}} \label{app_optimum_prediction}
Let us write the solution of the optimization problem \eqref{eqn_optimization_problem} at time $k+1$ in terms of the solution at time $k$ in the following integral form
\begin{equation}
\bbz^\star_{k+1} = \bbz^\star_k + \int_{0}^1\nabla_{\bbx}\bbz^\star(\bbx(k)+\theta\Delta\bbx)\Delta \bbx d\theta,
\end{equation}
where the expression of $\nabla_{\bbx}\bbz^\star(\bbx(k)\bbx)$ is that derived in \eqref{eqn_derivative_zstar} and $\Delta\bbx := \bm{\phi}(k+1,\bbx(0),\bbz_0)-\bm{\phi}(k,\bbx(0),\bbz_0)$. Thus, using the definition of $\bbz^\star_{k+1|k}$ (cf., \eqref{eqn_optimum_prediction}) it follows that
\begin{align}
\bbz^\star_{k+1|k}-\bbz^\star_{k+1}& = \nabla_{\bbx}\bbz^\star(\bbx(k))\Delta \bbx \nonumber \\&- \int_{0}^1\nabla_{\bbx}\bbz^\star(\bbx(k)+\theta\Delta\bbx)\Delta \bbx d\theta.
\end{align}
Notice that the previous expression can be also written as 
\begin{align}
\bbz^\star_{k+1|k}&-\bbz^\star_{k+1} = \\
&\int_{0}^1\left( \nabla_{\bbx}\bbz^\star(\bbx(k)) - \nabla_{\bbx}\bbz^\star(\bbx(k)+\theta\Delta\bbx)\right)\Delta \bbx d\theta,\nonumber
\end{align}
and we can upper bound its norm by 
\begin{align}\label{eqn_new_proof_diff}
&\left\|\bbz^\star_{k+1|k}-\bbz^\star_{k+1}\right\| \leq \\
&\int_{0}^1\left\| \nabla_{\bbx}\bbz^\star(\bbx(k)) - \nabla_{\bbx}\bbz^\star(\bbx(k)+\theta\Delta\bbx)\right\|\left\|\Delta \bbx\right\| d\theta.\nonumber
\end{align}
Using that the function $\nabla_{\bbx}\bbz^\star(\bbx)$ is Lipschitz with constant $2L\delta_1/m$ (see Lemma \ref{lemma_aux1}), the previous expression can be further upper bounded by 
\begin{equation}
\left\|\bbz^\star_{k+1|k}-\bbz^\star_{k+1}\right\| \leq \int_{0}^12\delta_1\frac{L}{m}\left\|\Delta \bbx\right\|^2\theta d\theta = \frac{L}{m}\delta_1 \left\|\Delta \bbx\right\|^2.
\end{equation}
The latter along with the assumption that $\left\|\Delta \bbx\right\|\leq BT_s$ completes the proof of the result. 
\subsection{Proof of Proposition \ref{prop_error}} \label{app_prop_error}
Since $\left\|\bbz_k-\bbz_k^\star\right\| \leq m/L$, Lemma \ref{lemma_bound_gradient} guarantees that the matrix $\nabla^2_{\bbz \bbz}\ccalL(\bbx(k),\bbz_k)$ is invertible, which guarantees that the prediction step is well defined. Then, consider the difference between the predicted iterate $\bbz_{k+1}^0$ and $\bbz^\star_{k+1}$, the solution of the problem \eqref{eqn_optimization_problem} at time $k+1$. Adding and subtracting $\bbz_{k+1|k}^\star$ to this difference yields
\begin{equation}
    \bbz_{k+1}^0-\bbz^\star_{k+1} = \bbz_{k+1}^0-\bbz^\star_{k+1|k}-\left(\bbz_{k+1}^\star-\bbz^\star_{k+1|k}\right),
\end{equation}
and use the definition of the prediction step \eqref{eqn_prediction} and that of $\bbz^\star_{k+1|k}$ [cf., \eqref{eqn_optimum_prediction}] to write the difference $\bbz_{k+1}^0-\bbz^\star_{k+1}$ as
  \begin{equation}
    \begin{split}
    \bbz_{k+1}^0-\bbz^\star_{k+1}     = \left(\bbz_{k}-\bbz^\star_{k}\right)-\left(\bbz_{k+1}^\star-\bbz^\star_{k+1|k}\right)\\
-\left(\nabla_{\bbz \bbz}^2\ccalL(k)^{-1}\nabla^2_{\bbz \bbx}\ccalL(k) - \nabla_{\bbz \bbz}^2\ccalL^\star(k)^{-1}\nabla^2_{\bbz \bbx}\ccalL^\star(k)\right)\Delta \bm{\phi}_{\kappa}.
\end{split}
    \end{equation}
  By virtue of the triangle inequality we can upper bound the norm of $\bbz_{k+1}^0-\bbz_{k+1}^\star$ by
\begin{equation}\label{eqn_zero_bound_prop_error}
    \begin{split}
&    \left\|\bbz_{k+1}^0-\bbz^\star_{k+1}\right\| \leq \left\|\bbz_{k}-\bbz^\star_{k}\right\|+\left\|\bbz_{k+1}^\star-\bbz^\star_{k+1|k}\right\|+ \\
&\left\|\nabla_{\bbz \bbz}^2\ccalL(k)^{-1}\nabla^2_{\bbz \bbx}\ccalL(k) - \nabla_{\bbz \bbz}^2\ccalL^\star(k)^{-1}\nabla^2_{\bbz \bbx}\ccalL^\star(k)\right\|\left\|\Delta \bm{\phi}_{\kappa}(k)\right\|.
\end{split}
    \end{equation}
   We focus next in bounding the norm of $\nabla_{\bbz \bbz}^2\ccalL(k)^{-1}\nabla^2_{\bbz \bbx}\ccalL(k) - \nabla_{\bbz \bbz}^2\ccalL^\star(k)^{-1}\nabla^2_{\bbz \bbx}\ccalL^\star(k)$. Adding and subtracting $\nabla_{\bbz \bbz}^2\ccalL^\star(k)^{-1}\nabla^2_{\bbz \bbx}\ccalL(k)$ yields
  \begin{equation}
    \begin{split}
      \nabla_{\bbz \bbz}^2\ccalL(k)^{-1}\nabla^2_{\bbz \bbx}\ccalL(k) - \nabla_{\bbz \bbz}^2\ccalL^\star(k)^{-1}\nabla^2_{\bbz \bbx}\ccalL^\star(k)\\
      =\left(\nabla_{\bbz \bbz}^2\ccalL(k)^{-1}- \nabla_{\bbz \bbz}^2\ccalL^\star(k)^{-1}\right)\nabla^2_{\bbz \bbx}\ccalL(k)\\
+ \nabla_{\bbz \bbz}^2\ccalL^\star(k)^{-1}\left(\nabla^2_{\bbz \bbx}\ccalL(k)-\nabla^2_{\bbz \bbx}\ccalL^\star(k)\right).
    \end{split}
    \end{equation}
  Take the norm in both sides of the previous equality. By virtue of the triangle inequality and the bounds in Assumptions \ref{assumption_non_zero_eigenvalues} and \ref{assumption_bounded_derivatives} we can write
  \begin{equation}\label{eqn_first_bound_prop_error}
    \begin{split}
    \left\|  \nabla_{\bbz \bbz}^2\ccalL(k)^{-1}\nabla^2_{\bbz \bbx}\ccalL(k) - \nabla_{\bbz \bbz}^2\ccalL^\star(k)^{-1}\nabla^2_{\bbz \bbx}\ccalL^\star(k)\right\|\\
      \leq C\left\|\nabla_{\bbz \bbz}^2\ccalL(k)^{-1}- \nabla_{\bbz \bbz}^2\ccalL^\star(k)^{-1}\right\|\\
+\frac{1}{2m}\left\|\nabla^2_{\bbz \bbx}\ccalL(k)-\nabla^2_{\bbz \bbx}\ccalL^\star(k)\right\|.
    \end{split}
    \end{equation}
  Observe that it is possible to write the first norm in the right hand side of the previous equation as
  \begin{equation}\label{eqn_need_for_space}
    \begin{split}
     \left\|\nabla_{\bbz \bbz}^2\ccalL(k)^{-1}- \nabla_{\bbz \bbz}^2\ccalL^\star(k)^{-1}\right\| \\
     = \left\|\nabla_{\bbz \bbz}^2\ccalL(k)^{-1}\left(\nabla_{\bbz \bbz}^2\ccalL^\star(k)- \nabla_{\bbz \bbz}^2\ccalL(k)\right)\nabla_{\bbz \bbz}^2\ccalL^\star(k)^{-1}\right\| \\
     \leq \frac{1}{2m^2} \left\|\nabla_{\bbz \bbz}^2\ccalL^\star(k)- \nabla_{\bbz \bbz}^2\ccalL(k)\right\|,
\end{split}
  \end{equation}
  where we have used that the minimum eigenvalues of $\nabla_{\bbz \bbz}^2\ccalL^\star(k)$ and $\nabla_{\bbz \bbz}^2\ccalL(k)$ are $2m$ and $m$ respectively. We can further upper bound the previous norm difference using the Lipschitz assumption of the second derivatives in Assumption \ref{assumption_bounded_derivatives} by
\begin{equation}\label{eqn_second_bound_prop_error}
    \begin{split}
      \left\|\nabla_{\bbz \bbz}^2\ccalL(k)^{-1}- \nabla_{\bbz \bbz}^2\ccalL^\star(k)^{-1}\right\| \leq \frac{L}{2m^2} \left\|\bbz_k-\bbz_k^\star\right\|.
\end{split}
  \end{equation}
  Likewise, we can upper bound the norm of the difference $\nabla^2_{\bbz \bbx}\ccalL(k)-\nabla^2_{\bbz \bbx}\ccalL^\star(k)$ can be upper bounded by
  \begin{equation}\label{eqn_third_bound_prop_error}
    \begin{split}
      \left\|\nabla^2_{\bbz \bbx}\ccalL(k)-\nabla^2_{\bbz \bbx}\ccalL^\star(k)\right\| \leq  L\left\|\bbz_k-\bbz_k^\star\right\|.
    \end{split}
  \end{equation}
  {
  Combining the bounds \eqref{eqn_first_bound_prop_error}, \eqref{eqn_second_bound_prop_error} and \eqref{eqn_third_bound_prop_error} yields 
  \begin{equation}\label{eqn_fourth_bound_prop_error}
    \begin{split}
    \left\|\nabla_{\bbz \bbz}^2\ccalL(k)^{-1}\nabla^2_{\bbz \bbx}\ccalL(k) - \nabla_{\bbz \bbz}^2\ccalL^\star(k)^{-1}\nabla^2_{\bbz \bbx}\ccalL^\star(k)\right\| \\
    \leq \frac{L}{2m}\left(\frac{C}{m}+1\right)\left\|\bbz_k-\bbz_k^\star\right\| = \frac{L}{m}\delta_2\left\|\bbz_k-\bbz_k^\star\right\|.
\end{split}
  \end{equation}
  Substituting in \eqref{eqn_zero_bound_prop_error} the bound derived in \eqref{eqn_fourth_bound_prop_error} yields the following upper bound for $\left\|\bbz_{k+1}^0-\bbz^\star_{k+1}\right\| $
    \begin{equation}
    \begin{split}
      \left\|\bbz_{k+1}^0-\bbz^\star_{k+1}\right\| \leq \left\|\bbz_{k+1}^\star-\bbz^\star_{k+1|k}\right\|\\
      +\left(1+\frac{L}{m}\delta_2\left\|\Delta \bm{\phi}_{\kappa}(k)\right\|\right)\left\|\bbz_{k}-\bbz^\star_{k}\right\| 
\end{split}.
    \end{equation}
}
The result of Proposition \ref{prop_optimum_prediction} completes the proof.